\documentclass[12pt]{article}
\pdfoutput=1
\usepackage[margin=1in]{geometry}
\usepackage{amsmath,amssymb,amsthm}
\usepackage{mathtools}
\usepackage{stmaryrd}

\usepackage{epsfig}
\usepackage{mdframed}
\usepackage{xr}
\usepackage{xr-hyper}
\usepackage[colorlinks = true]{hyperref}
\usepackage{xcolor}
\definecolor{darkred}  {rgb}{0.5,0,0}
\definecolor{darkblue} {rgb}{0,0,0.5}
\definecolor{darkgreen}{rgb}{0,0.5,0}
\hypersetup{
  urlcolor   = blue,         
  linkcolor  = darkblue,     
  citecolor  = darkgreen,    
  filecolor  = darkred       
}
\usepackage[utf8]{inputenc}

\usepackage{cleveref}
\usepackage{todonotes}
\newtheorem{theorem}{Theorem}
\newtheorem*{theorem*}{Theorem}
\newtheorem*{remark*}{Remark}

\newtheorem{definition}[theorem]{Definition}
\newtheorem{algorithm}[theorem]{Algorithm}
\newtheorem{claim}[theorem]{Claim}
\newtheorem{lemma}[theorem]{Lemma}

\newtheorem{corollary}[theorem]{Corollary}

\def\Z{{\mathbb{Z}}}
\def\R{\mathbb{R}}
\def\C{\mathbb{C}}
\def\poly{{\rm poly}}
\def\bdd{{\rm BDD}}

\def\samplebdd{{\rm SampleBDD}}
\def\samplehip{{\rm SampleHIP}}

\newcommand{\smfrac}[2]{\mbox{$\frac{#1}{#2}$}}

\newcommand\ket[1]{{ |{#1} \rangle }}

\def\dist{{\rm dist}}

\newcommand{\om}{\omega}

\newcommand{\phase}{\theta}
\newcommand{\phasesamp}{h}

\newcommand{\eva}[1][]{\ket{\psi_{#1}}}
\newcommand{\eps}{{\varepsilon}}

\newcommand{\pmr}{[\pm\rl]}
\newcommand{\ppower}{T}

\newcommand{\zq}{\Z_q}

\newcommand{\zqn}{\Z_q^n}

\newcommand{\matB}{\mathbf B}

\newcommand{\G}{G} 
\newcommand{\matG}{\mathbf G} 
\newcommand{\matI}{\mathbf I}
\newcommand{\matR}{\mathbf R}
\newcommand{\matS}{\mathbf S}
\newcommand{\coeffs}{{\tilde{C}}}

\newcommand{\va}{\mathbf a}
\newcommand{\vb}{\mathbf b}
\newcommand{\vc}{\mathbf c}
\newcommand{\vd}{\mathbf d}
\newcommand{\ve}{\mathbf e}
\newcommand{\vg}{\mathbf g}
\newcommand{\vq}{\mathbf q}
\newcommand{\vr}{\mathbf r}
\newcommand{\vs}{\mathbf s}
\newcommand{\vt}{\mathbf t}
\newcommand{\vu}{\mathbf u}
\newcommand{\vv}{\mathbf v}

\newcommand{\vx}{\mathbf x}
\newcommand{\vy}{\mathbf y}
\newcommand{\vz}{\mathbf z}

\newcommand{\vzero}{{\mathbf{0}}}
\newcommand{\vD}{\mathbf \Delta}

\newcommand{\grank}{r}  
\newcommand{\gr}{r}  
\newcommand{\sampsp}{\zq^\grank}

\newcommand{\rl}{{\sigma}}
\renewcommand{\sl}{{2\rl}}

\newcommand{\groupdecomp}{(\matG,\vq,\grank)}

\newcommand{\charval}[2]{\chi_{#1}(#2)}
 \newcommand{\charphase}[2]{f(#1,#2)}
 \renewcommand{\charphase}[2]{#2 \cdot #1}

\newcommand{\fgr}{\text{finite-group-rank}}

\newcommand{\rndL}{{\tilde{L}}}
\newcommand{\rndlambda}{\tilde{\lambda}}

\newcommand{\rndmatB}{{\tilde{\matB}}}
\newcommand{\rndG}{\tilde{G}}
\newcommand{\rndmatG}{{\tilde{\matG}}}

\newcommand{\pcs}[1]{\ket{\psi_{#1}}}
\newcommand{\cube}[1]{\ket{C(#1)}}  
\newcommand{\cubep}[1]{\ket{C(#1)}}  
\newcommand{\cubet}[1]{\langle C(#1) }  

\newcommand{\fbdd}{\eps_1}
\newcommand{\fbddr}{\tilde{\eps}_1}  
\renewcommand{\>}{\rangle}
\newcommand{\<}{\langle}
\newcommand{\shape}{C}
\newcommand{\dimin}{n} 
\newcommand{\dimout}{m} 

\newtheorem{innercustomthm}{Theorem}
\newenvironment{customthm}[1]
  {\renewcommand\theinnercustomthm{#1}\innercustomthm}
  {\endinnercustomthm}

\begin{document}

 \title{An efficient quantum algorithm for lattice problems 
   achieving subexponential approximation factor }

 \author{Lior Eldar\thanks{eldar.lior@gmail.com}
   \and
   Sean Hallgren\thanks{
     Department of Computer Science and Engineering, Pennsylvania State
   University, 350W Westgate Building, University Park, PA  16802,
   USA.  Partially supported by National Science
   Foundation awards CNS-2001470 and OIA-2040667, and by a
   Vannevar Bush Faculty Fellowship from the US Department of Defense.
   This work was done in part while visiting the Simons Institute for the
   Theory of Computing.}
}

\date{}
 \maketitle 

 \vspace{-.2in}
   
\begin{abstract}
 We give a quantum algorithm for solving the Bounded
 Distance Decoding (BDD) problem with a subexponential approximation factor
 on a class of integer lattices.  The quantum algorithm uses a
 well-known but challenging-to-use quantum state on lattices as a type of
 approximate quantum eigenvector to randomly self-reduce the BDD
 instance to a random BDD instance which is solvable classically.  The
 running time of the quantum algorithm is polynomial for one range of 
 approximation factors and subexponential time for a second range of
 approximation factors.

 The subclass of lattices we study has a natural description in terms of the
 lattice's periodicity and finite abelian group rank.  This view makes
 for a clean quantum algorithm in terms of finite abelian groups, uses
 very relatively little from lattice theory, and suggests exploring approximation
 algorithms for lattice problems in parameters other than dimension
 alone.  

 A talk on this paper sparked many lively discussions and resulted in
 a new classical algorithm matching part of our result.  We leave it
 as a challenge to give a classcial algorithm matching the general case.
\end{abstract}

\section{Introduction}

We give an efficient quantum algorithm for a special case of the closest
lattice vector problem in a new range of approximation
factors, namely the subexponential range.  In this type of problem a basis $\matB \in \Z^{n \times n}$
and a target vector $\vt\in \Z^n$ are given and the goal is to
compute the closest lattice vector $\matB\vc$ to $\vt$ for integer
coefficients.  This
paper is about the important special case called Bounded Distance
Decoding (BDD) with parameter $\fbdd$.  It has the extra promise that
the distance is bounded in the sense that there exists $\matB\vc\in L$
such that $\|\matB\vc- \vt\| < \fbdd \lambda_1$, where $\lambda_1$
is the shortest nonzero vector length in $L$, and $\fbdd \leq 1/2$,
making the answer unique.  The term $\fbdd$ is the approximaton factor
and it is typically a function of $n$, the lattice dimension.  A
lattice can be specified by an infinite number of bases,
making the problem difficult.

There are three broad ranges
of approximation factors where lattice 
problems are unlikely to be NP-complete.
\begin{figure}[h]
  \label{fig:pic-ranges}
  \begin{center}
    \includegraphics[width=5.0in]{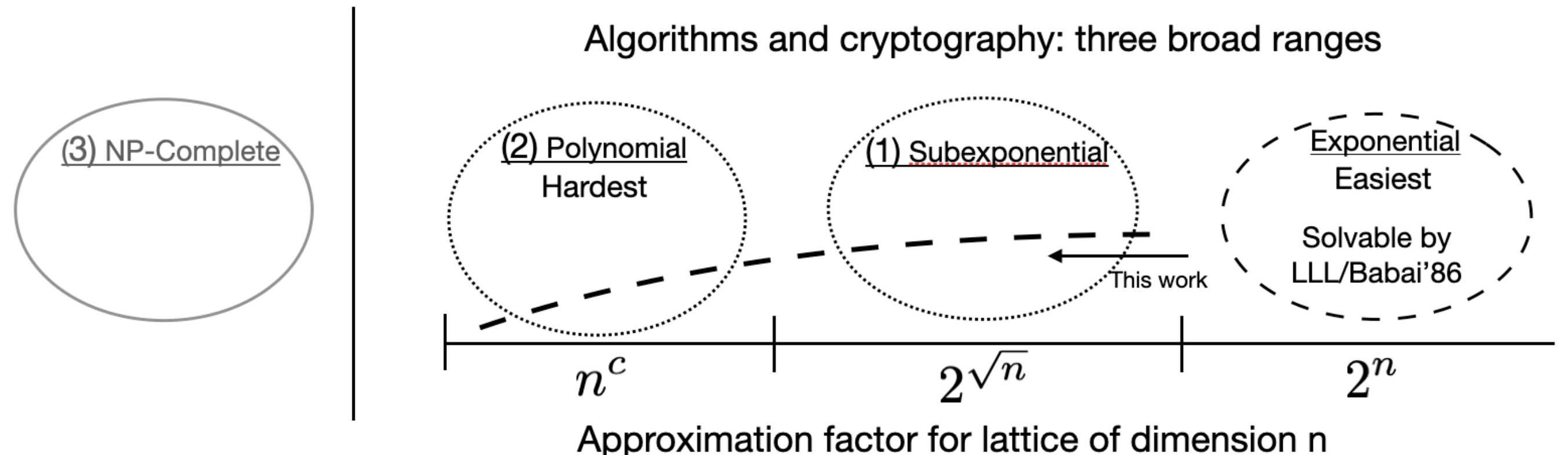}\quad \quad 
  \end{center}
\end{figure}
 Starting at the right end, the exponential range typically has
the 
form $\eps_1^{-1} = 2^n$ and was solved efficiently in the 80's.  Lentra, Lenstra, and
Lov\'{a}sz~\cite{LLL82} gave an algorithm
to compute an approximate shortest lattice vector, Babai~\cite{Bab86}
gave algorithms for computing an approximate closest lattice vector to
a target point, Kannan~\cite{Kan87}
gave an exponential time
enumeration algorithm, and
Schnorr~\cite{S87,S94} extended these three by trading off
running time  in exchange for better approximation factors.

Adjacent to the exponential range is the subexponential range, which
is the focus of this paper.   Despite several decades of big advances in
lattices, this region appears to either be very difficult, or has been
neglected.  In special-case lattices that allow more more efficient
cryptography because of additional algebraic structure, a sequence of
papers work led to an efficient quantum algorithm for approximating
the shortest vector~\cite{EHKS14,BS16,CGS14,CDPR16}.  Lattices with
small determinant have also been examined~\cite{CL15}.  The subexponential region
has played a crucial role in recent advances in fully homomorphic 
encryption (FHE)~\cite{BV11b,BV11a,GSW13}, where it is assumed that certain
parameters cannot be solved efficiently.

At the hardest end of the range for algorithms and cryptography are
polynomial approximation factors and very important questions about
how well existing algorithms work and can be optimized for concrete
security for the NIST standardization process.

Cryptography built on the LWE problem, which  
is as hard as worst-case lattice problems for good theoretical
security, also allows many new primitives such as FHE
and testing whether or not machines are  
quantum~\cite{BCMVV18, Mah18b},
and quantum FHE \cite{Mah18a}.

Therefore, the most important and pressing question is whether or
not efficient algorithms exist for the polynomial
approximation factor range.  The
realistic approach is to start with problems in the range adjacent to the ones that
are already solvable, which is the subexponential range.  Even if an
efficient algorithm exists for the polynomial range it may be too
difficult to find in one step, but the hope is that the techniques
here will be applicable to a broader range of cases.

We propose a partition of lattices into finer 
 blocks than just by dimension alone.
The partition consists of sets of lattices ${\cal L}(n,q,r)$ indexed
by lattice dimension $n$, periodicity $q$, and finite group rank
$\grank$.  The periodicity $q$ of a lattice is the minimum integer $q$
such that lattice contains the subgroup $q\cdot \Z^n\subseteq L$.
Therefore $L \bmod q$ is a finite abelian subgroup of $\Z_q^n$ and can
be decomposed as $\Z_{q_1} \times \cdots \times \Z_{q_\grank}$. 

Through this lens we give a quantum algorithm on a subset
of lattices with parameter $\grank \log q$ achieving a subexponential
approximation factor of the form $2^{-\sqrt{\grank \log q}}$ and
running in time $\poly(n,\log q)$
(Theorem~\ref{thm:main1}).  To give 
a simplified comparison to existing algorithms analyzed in terms of dimension, for
example, lattices with $\sqrt{\grank \log q }< n$ we get an improvement
over Babai's algorithm, and BDD on lattices with finite group rank
$\gr=n^{1/4}$ and periodicity $q=2^{\sqrt{n}}$ can be solved for
approximation factor $2^{-n^{3/8}}$.  The periodicity $q$ cannot be
smaller than $\gr$ or BDD becomes trivial.  More generally, the new
parameter range can be 
seen in the figures for $r=1$ and general $r$.  Changing $\gr$ changes
the trivial region.  The left axis has the
parameter $\gr \log q$ and the bottom axis plots the log of the
approximation factor. 

\begin{figure}[h]
  \label{fig:rvalues}
  \begin{center}
  \includegraphics[width=2.0in]{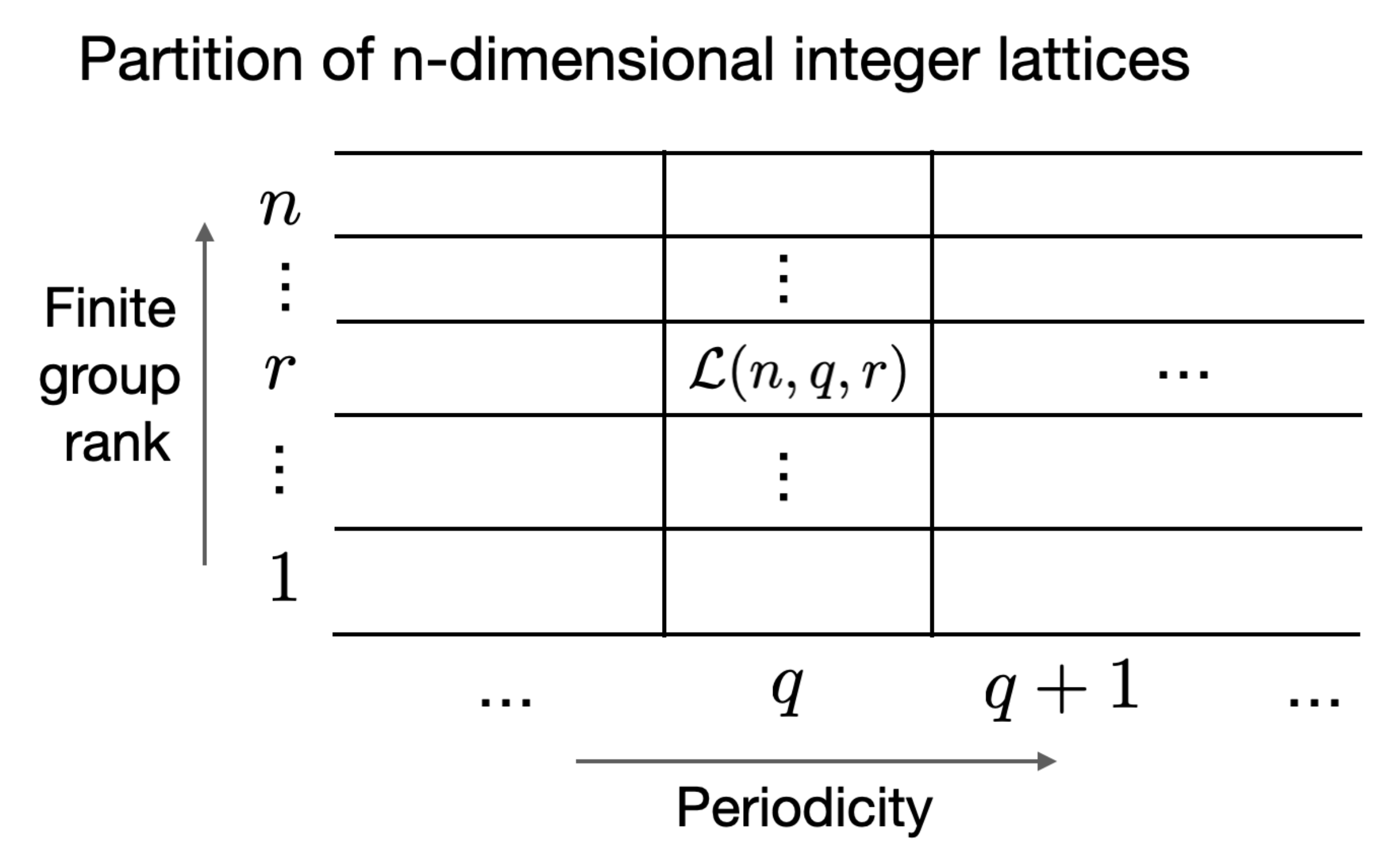}\quad
    \includegraphics[width=2.0in]{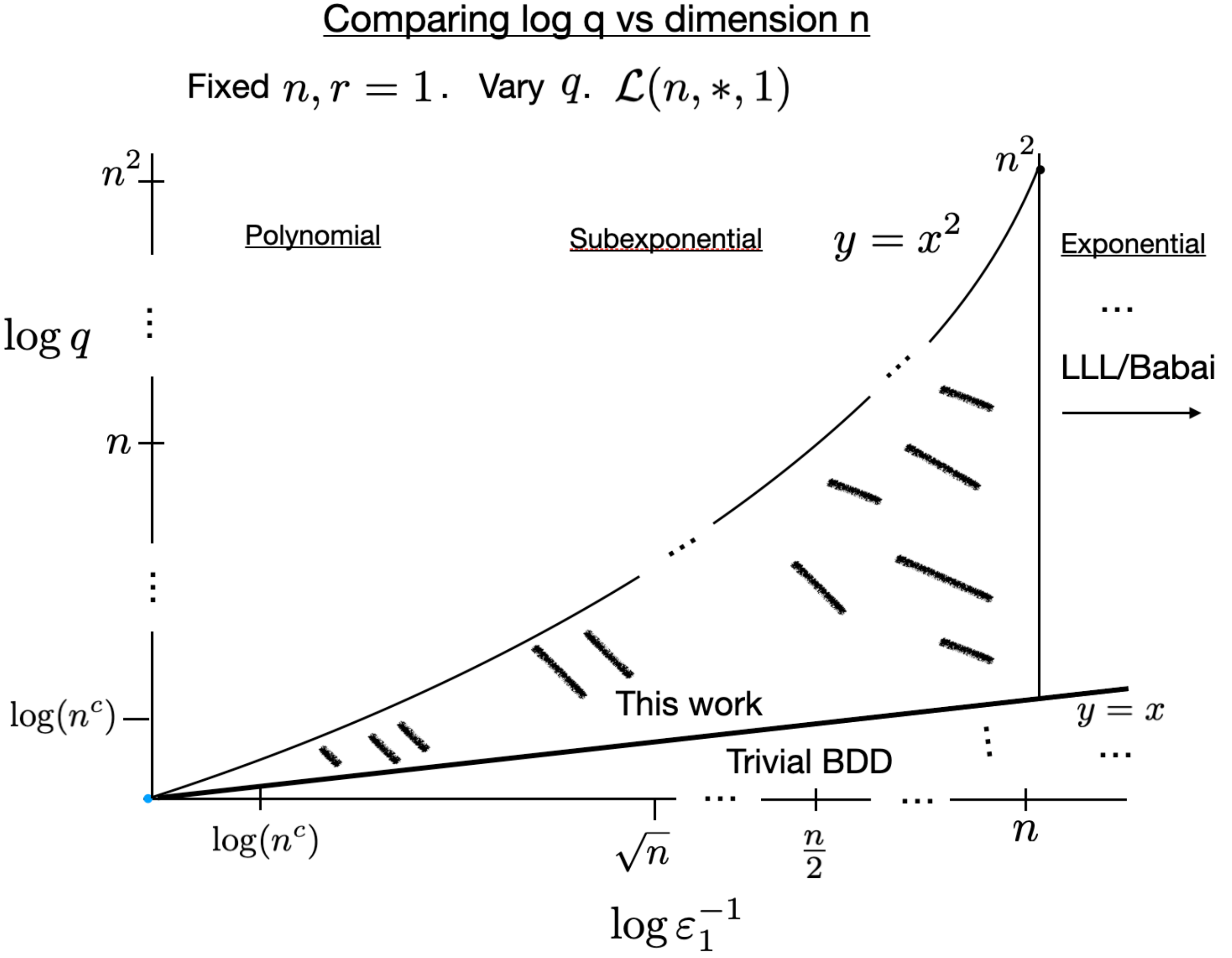}\quad 
  \includegraphics[width=2.0in]{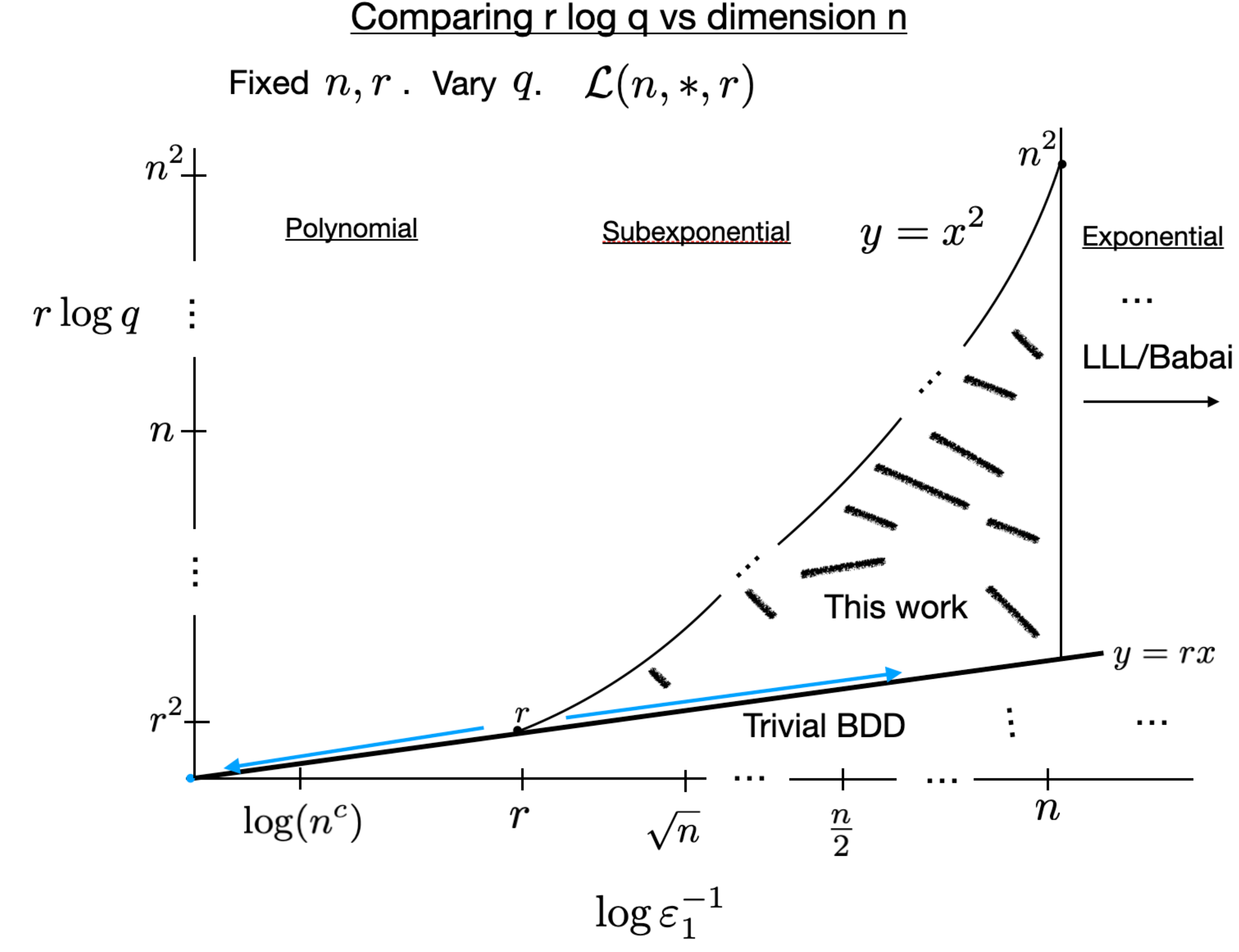}
\end{center}
\end{figure}

Moving beyond polynomial time, Schnorr's hierarchy can also be applied
to the output of the quantum algorithm as a black box to achieve an
approximation factor of $\exp(-\frac{ \gr( \log q)\log \beta}{\beta})$
in time $\beta^\beta \poly(n,\log q)$.  For example, with
$\eps \leq 1/2$, there is a quantum algorithm achieving approximation
factor $2^{n^\eps \log n}$ in time $2^{n^{1-2\eps}}$, while
applying Schnorr to the original lattice
gives time $2^{n^{1-\eps}} $, i.e., with an
extra $n^\eps$ in the time exponent.  

{\bf Idea of proof.} In the context of the partition, we reduce BDD on worst-case lattices
in ${\cal L}(n,q,\gr)$, to a problem we'll call
\emph{$\fbddr$-random-BDD} where a
  random matrix $\rndmatB \in [q]^{m\times \gr}$ is chosen and $\rndL
  \in {\cal L}(m,q,r)$
  has as generators the columns of $[\matB| q\cdot \matI]$, and a
  target vector $\tilde{\vt}$ is given with distance at most
  $\fbddr\rndlambda_1$ to $\rndL$.  From lattice theory we
  need that these types of lattices have a
  long shortest vector with high probability~\cite{Mic19notes}.  The
  problem can be solved for approximation factor 
  $\fbdd= 2^{-\sqrt{\gr \log q}}$ when the dimension is $m=\sqrt{\gr \log q}$
  by running LLL and Babai's algorithm.  Different from the worst-case
  BDD problem, the random-BDD problem can be solved for higher
  dimensions too because  deleting rows until $\sqrt{\gr \log q}$ is
  reached is still a random instance.  

To randomly reduce the BDD instance to a random-BDD instance in a lower
dimension, we revisit an old state with a ``phase problem'' and find a
way to use it.  Given a lattice basis and a radius, a goal for the last 20+ years has
been to compute the quantum state
$\pcs{\vzero} =\sum_{\vv\in L} \sum_{\vz \in \shape} \ket{\vv+\vz}$,
where $\shape$ is some shape with the prescribed radius
such as a cube, Gaussian, or sphere (also see~\cite{AR05}).  We use a
cube for simplicity.

The main approach for computing $\pcs{\vzero}$ is to compute a
superposition over
coefficients $\vc$ in the first register and over the cube in the second
register to get $\sum_{\vc} \ket{\vc} \otimes \sum_{\vz\in \shape}
\ket{\vz}$, then to entangle the registers by adding the 
corresponding lattice vector into the
second register to get $\sum_{\vc} \sum_{\vz\in \shape} \ket{\vc,
  \matB\vc+\vz}$.  The last step would be to ``uncompute'' in the
first register and this is where the state becomes difficult to use.  If an algorithm could solve BDD
at this point 
then it could use $\matB\vc+\vz$ to compute $\vc$ and
uncompute the first register, but this is circular.

This challenge of uncomputing a register appears in lattice problems,
graph isomorphism approaches, as well as other problems~\cite{AT07}.
Nevertheless, Regev~\cite{Reg09} found a way to use this in a
constructive way by using an LWE oracle to erase the coefficient, and 
as a result of his construction, reduce worst-case lattice problems to
LWE.  The more typical approach for algorithms might be to compute the
quantum Fourier transform of the first register, because it is
possible, and measure it. 
The resulting state is
$\pcs{\va} = \sum_{\vc} \om_q^{\vc \cdot \va} \sum_{\vz\in \shape}
\ket{\matB \vc+\vz}$, for a random vector $\va$.
This changes the ``uncomputing'' problem into a
``phase'' problem because
$\pcs{\va}$ is similiar to $\pcs{\vzero}$, the desired $\va=\vzero$ case, but has the phase
$\om_q^{\vc\cdot \va}$ mixed in in a problematic way.  The difficulty
is that it is no
longer clear how to use $\pcs{\va}$.  In~\cite{BKSW18} states
related to this but with Gaussians were used to show an equivalence
between LWE and an extension of a certain nonabelian hidden subgroup problem.  In
\cite{CLZ21} the Arora-Ge
algorithm~\cite{AG10} for
LWE  was used to uncompute projections.

{\bf Overcoming the difficulty.}
In this paper we revisit the state $\pcs{\va} =
\sum_{\vc} \om_q^{\vc \cdot \va} \sum_{\vz\in \shape}
\ket{\matB \vc+\vz}$ which we call a Phased Cube State (PCS), because
there is a cube around each lattice point, and each cube has single
phase across it.  The value $q$ is chosen as the periodicity of the
lattice, and $L\bmod q$ is a subgroup $\Z_q^n$ where computations are done.
This state can be 
efficiently created for any side length, and  with a  uniformly random 
and known $\va$.  The main idea is to see that $\pcs{\va}$ is almost
an eigenvector of shifts by vectors close to the lattice, as in BDD
instances, and that quantum phase estimation can compute an
approximation of that information.  More specifically, for
$\vx \in \Z_q^n$, let $U_\vx \ket{\vy}=\ket{\vy+\vx}$ be a shift
operator inside $\Z_q^n$.  Then for $\matB \vc\in L$,
$U_{\matB\vc} \pcs{\va} = \om_q^{\va \cdot \vc}\pcs{\va}$.  If
desired, the quantum phase estimation algorithm can be used to compute the inner
product $\va\cdot \vc$, and repeating the process results in these inner
products for different random $\va$, and the coefficients $\vc$ can be computed.

For a BDD instance with $\vt=\matB \vc+\mathbf\Delta$ and
$\mathbf \Delta$ controlled by the BDD promise, shifting by $\vt$
results in
$U_\vt \pcs{\va} = \om_q^{\vc\cdot \va} U_{\mathbf \Delta}
\pcs{\va} \approx \om_q^{\vc\cdot \va}\pcs{\va} $ and we show that
quantum phase estimation still returns an approximation of
$\vc\cdot \va$.  Quantum phase estimation exponentiates the operator
to the power of the precision requested, and because of degradation of
$\pcs{\va}$ for higher powers, because $U_\vt^k \pcs{\va} =
\om_q^{k \va \cdot \vs} U_{k\mathbf \Delta} \pcs{\va}$, this limits how
much information can be 
extracted.  To accomodate the set of possible BDD target vectors we
define the notation of having a set of operators together with a
single approximate eigenvector.  With the worst-case lattice problem
BDD as input, this quantum subroutine is used to sample noisy inner
products and construct a random BDD instance in a lower solvable
dimension.

To summarize, we give a new quantum algorithm solving BDD on a range
of subexponential approximation factors.

\subsection{ Comparison to classical algorithms and open problems}
The case of $\gr=1$ and exponential $q$ includes the well-studied
Hidden Number Problem~\cite{BV96,Aka09} which appears to need
structure beyond the worst-case to solve until this work.

\begin{remark*}
  A talk given on this work in September 2021 at the Simons Institute
  for the Theory of Computing sparked much productive discussion.  In
  particular, the paper~\cite{DW21} was posted with a classical algorithm solving a
  part of what we solve that had not appeared in the literature
  before.  After going through existing literature, having discussions
  with the community, and waiting for responses, we leave it as a
  challenge to provide a classical algorithm matching  the subexponential-time
  quantum algorithm.
\end{remark*}

For the polynomial-time range, we show
\begin{customthm}{\ref{thm:main1}}
  There is a $\poly(n, \log q)$-time quantum algorithm solving 
  $2^{-\Omega(\sqrt{\grank \log q})}$-BDD on lattices of dimension 
  $n$, periodicity $q$, and finite  group rank $\grank$. 
\end{customthm}

This inspired the posting of a classical algorithm for this problem.
In that paper~\cite{DW21}, the $r=1$ case matches Theorem 20, but
details are left out about how to match the $r>1$ case.  In
Section~\ref{sec:lll} we complete the analysis and also generalize it
to rectangle-periodic lattices, rather than just cubes.  Our quantum
algorithm already handles this type of lattice because it works for
any finite group, so classical and quantum have the same performance,
even though the ideas are completely different.  The paper does not
address our next theorem, which is exponentially faster than the best
available classical algorithm:

\begin{customthm}{\ref{thm:main}}
   Let $L(\matB)$ be an $n$-dimensional $q$-periodic lattice with
   finite group rank $\grank$.
   Given an instance of $\eps_1$-BDD $(\matB, \vt)$, with $2\leq
   \beta \leq \grank\log q$, and
  $$
  \eps_1 = \left(\exp\Big(\!- 4\sqrt{\smfrac{\grank \log q \log
        \beta}{\beta}} \Big)
    \cdot 2^2 m p_{\ref{lem:samplebdd}}(n,\log q)^2 \right).
  $$ 
  Algorithm \ref{alg:tradeoff} runs
  in quantum time $\approx \beta^\beta \poly(n, \log(q))$ and returns the closest vector to $\vt$ with probability at least $0.9$.      
\end{customthm}

One other suggestion for a classical approach for this problem was
using~\cite[Theorem 5.3]{GMPW20}.  It is not clear how 
the details would work, and in particular, how to handle the
non-primitive case.  

It is still open if the Schnorr trade-off we give in
Section~\ref{sec:schnorr} can be done classically.  These questions are out-of-scope for
this paper and we leave them as open problems.  The parameter
$k \log q$ for a matrix of dimension $k\times n$ has appeared as a
boundary for LWE where algorithms more carefully use the Gaussian
error~\cite{MR09,LP11, BLPRS13,BCMVV18}.

There are many open problems and possible extensions.  The most
interesting is to try variations of the quantum algorithm.   It is
relatively clean and is easy to experiment with.  For example,
solving approximate arithmetic progressions is a posssibility.
Another is analyzing
different worst-case lattice problems such as uSVP, using reductions
between instances of LWE, which is a type of random BDD where the
errors for each coordinate are i.i.d., to map between different
dimensions and $q$ values, and use groups~\cite{BLPRS13,GINX16}.
These new BDD algorithms can also be used to sample
vectors
of length $q/(\fbdd\lambda_1) \geq \eta(L^\perp)
2^{\sqrt{r \log q}}$ in $L^\perp$.  This
can be done via quantum (\cite[Theorem 1.3]{Reg09}) or classical sampling (\cite{GPV08}).
\nocite{kup05}

\section{Background}\label{sec:def}

\subsection{Lattices, finite abelian groups, and distances}
\label{sec:lat-group}

Every integer lattice in $\Z^n$ has minimum $q$ such that
$q \Z^n \subseteq L$, and $L$ is called \emph{$q$-periodic}, or
\emph{$q$-ary}.  Because $q\Z^\dimin$ is a subgroup of $L$,
$L \bmod q := L/q\Z^n$ has all information about the lattice in the
sense that distances are preserved mod $q$ and
$L = (L\cap [q]) + q\Z^n$.  Computing the closest vector to a lattice
over $\Z^\dimin$ can be reduced to this case by reducing the lattice
and vector mod $q$, solving the problem in $L \bmod q$, and then
mapping back to the integer solution of the original problem.
Starting from the finite group the associated lattice is constructed
by reintroducing $q \Z^\dimin$, so the columns of
$[\matB,q\cdot\matI]$ generate $L$.

A finite abelian group can be decomposed as 
$L \bmod q \cong \Z_{q_1} \times \cdots \times \Z_{q_\grank}$.  The
representation is the
\emph{finite group decomposition} $\groupdecomp \in  \zq^{\dimin  
  \times \grank} \times \Z^\grank$, where the columns of
$\matG \in \Z_q^{\dimin \times \grank}$ span
$G := L \bmod q\leq \Z_q^\dimin$, the vector
$\vq = (q_1,\ldots, q_\grank)\in \Z^\grank$ gives the orders of each
column in $\matG$ in the decomposition, and $\grank$ makes the rank
visibile in the notation.  The specific decomposition can be chosen
but in this case the unique one will be used where
$q_i | q_{i+1}$.  The set $\coeffs :=
[q_1]\times \cdots \times [q_\grank]$ will be viewed a the set
of coefficients of the group elements $\vv \in \G$, where each has a unique coefficient
vector $\vc\in \coeffs$ such that
$\vv = \matG \vc$.

The setup is similar to the matrix $A\in \Z_q^{\dimin \times \dimout}$ used in
lattice-cryptography, but not exactly the same.  In the worst-case to
average-case reduction the input lattice $L$ has dimension $\dimin$
and $q$ is arbitrary subject to sampling in the dual.
The matrix $A$ is typically chosen randomly and as a result has finite
group decomposition 
$\Z_q^\dimin$ with high probability.  Here we also are using $q$-ary
lattices, but we start with a worst-case lattice $L$, use the specific
periodicity $q$ of $L$, and decompose it mod $q$.

An arbitrary full-dimensional integer lattice
$L \subseteq \Z^\dimin$ is a $q$-periodic lattice for $q=\det(L)$.
This can be seen from the fact that
$q B^{-1} \in \Z^{\dimin\times \dimin}$, because using Cramer's rule
for inverting $B$ results in each entry having $\det(L)$ in the
denominator an integer in the numerator.  Then using the integer
vectors from the columns of $qB^{-1}$ takes $B$ to $q I
= B(q B^{-1})$, which is $0 \bmod q$.
The parameters set this way may not always work in the quantum
algorithms, for example,
when $\det(L)$ is too large relative to $\dimin$.

Given a lattice $L \subseteq \Z^n$ the finite group decomposition can
be efficiently computed.

The quantum Fourier transform over the cyclic group $\Z_{q}$ maps
$\ket{c}$ to $\frac{1}{\sqrt{q}} \sum_{a=0}^{q-1} \omega_q^{ac}
\ket{a}$.  In general for a finite group $G$ the Fourier transform
maps vectors over the group to vectors over the character group
$\hat{G} = \{\chi_a: G \rightarrow \C^* : a \in G\}$.  

There will be a reindexing step for the eigenvector/eigenvalue
calculation when a register holding a superposition of coefficients
$\sum_{\vc \in \coeffs}\ket{\vc}$ is
transformed by the Fourier transform over $\Z_q^\grank$.  This uses
the subgroup embedding of $\coeffs$ into $\Z_q^\grank$.  Concretely
this means that $(c_1,\ldots,c_r)\in \coeffs$ maps to the element $(\frac{q}{q_1} c_1
,\ldots, \frac{q}{q_\grank}c_\grank) \in \Z_q^\grank$, and $\chi_\va(\vc)$
has phase $\frac{\sum_{i=1}^\grank \frac{q}{q_i}c_i}{q} = \sum_{i=1}^\grank \frac{c_i}{q_i} $.
Then
$\sum_{\vc \in \coeffs} \chi_\va(\vc) \ket{ \matG (\vc+\vd)}
=\sum_{\ve \in \coeffs} \chi_\va(\ve-\vd) \ket{ \matG \ve}
=\chi_\va(-\vd) \sum_{\ve \in \coeffs} \chi_\va(\ve) \ket{ \matG
  \ve}$.

A distance on $\zqn$ will be needed to define and solve BDD on
subgroups $\G$ of $\zqn$, and also for the phase estimation statement.
Following Cassels \cite{Cas97} specialized to finite groups, with
$\Lambda = q \Z^\dimin$, the modular distance on the quotient
$\zqn =\Z^\dimin/\Lambda = \Z^\dimin/(q\Z^\dimin)$ is
defined from the Euclidean distance on $\Z^\dimin$ by
$\|\vy\|_q = \min_{\va \in \vy+q \Z^\dimin}
\|\va\|$.  For any $\vy \in \zqn$, $\|\cdot \|_q$
satisfies (1) $\|k\vy\|_q \leq k \|\vy\|_q$ for
integers $k\geq 0$, (2)
$\|\vy + \vz\|_q \leq \|\vy\|_q +
\|\vz\|_q$ for $\vz \in \zqn$, (3) and there exists
$\va \in \vy+q\Z^\dimin$ such that
$\|\vy\|_q = \|\va\|$.  In one dimension we may write
$| y |_q$ for $\|y\|_q$.  The distance between points in
$\zqn$ matches the Euclidean distance as long as it is at most $q/2$.
This definition also allows any choice of coset representatives for
$\Z_q = \Z/q\Z$.  It is equal to the zero-centered set for $\Z_q$ where the class $y
\bmod q$
is represented by an integer $x$ so that $-q/2
\leq x \leq q/2$, then it holds that
$\vx\in \zqn$, $\| \vx \|_q = \| \vx\|$ when $\| \vx \| \leq 
\frac{q}{2}$, and $\| \vx \|_q = q - \| \vx \|$ when $\| \vx \| \geq \frac{q}{2}$.

For phase estimation we will also use a distance mod $1$.  In this
case take the Euclidean distance on $\R$ and define the
distance on $\R/\Z$ by $| y |_1 = \min_{a \in y + \Z} |a|$.  This has
the same properties listed above, but to $1/2$ instead of $q/2$.

Given a subgroup $\G$ of $\zqn$, define a shortest (nonzero) element
length to be 
$$
\lambda_1(\G) = \min_{\vv\in G\backslash\{\vzero\} } \|\vv\|_q.
$$  
Also
define $\dist_q(\vy,G) = \min_{\vv \in G} \| \vy -
\vv\|_q$ for any $\vy \in \zqn$.  Note that all group
elements have length at most $q/2$.  In particular, for the trivial case
when $\Lambda_q = q\Z^\dimin$, $\lambda_1(\Lambda_q)=0$.

The main tool with $q$-ary lattices in dimension $\dimout$ is that for a
randomly chosen one, the shortest vector length is known within a
constant with high probability.  The following can be found in \cite{Mic19notes}.
\begin{claim}\label{fact:randomq}
There exists a constant $\delta >0$ such that if $\rndmatG$ is a uniformly chosen 
matrix from $\Z_q^{\dimout \times \grank}$, 
and let $L(\rndmatG)$ denote the corresponding $q$-ary lattice.
Then
\begin{enumerate}
\item
$\Pr_\rndmatG( \lambda_1(L(\rndmatG)) < \delta \sqrt{\dimout} q^{1 -
  \grank/\dimout}) \leq 1/2^\dimout$,  
\item
$\Pr_\rndmatG (\Z_q^\dimout \rndmatG =\Z_q^\grank ) \geq 1 - 1/q^{\dimout-\grank}$.

\end{enumerate}

\end{claim}

To distinguish the underlying operations, ranges for integers will be
donoted by $[q] := \{0,1,\ldots,q-1\}$.  In this 
case, for example, addition and multiplication of numbers from $[q]$
are over $\Z$.  If an element $\vg \in \Z_q^m$, for example, then
addition is mod $q$.  An integer times a group element represents
the number of operations to perform, for example, for $\vg \in
\Z_q^m$, $3\in [q]$, $3\vg =
\vg+\vg+\vg$.


The following basis reduction algorithm will be used:
\begin{lemma}[{\cite[Lemma 7.1]{MG02}}]
\label{lem:MGalg}
There is a polynomial time algorithm that on input a lattice basis
$\matB$ and linearly independent lattice vectors $\matS \subset
L(\matB)$ such that $\|\vs_1\|\leq \|\vs_2\| \leq \cdots \leq
\|\vs_n\|$, outputs a basis $\matR$ equivalent to $\matB$ such that
$\|\vr_k\| \leq \max \{(\sqrt{k}/2)\|\vs_k\|,\|\vs_k\| \}$ for all
$k=1,\ldots, n$.  Moreover, the new basis satisfies
$\text{span}(\vr_1,\ldots, \vr_k) = \text{span}(\vs_1,\ldots,\vs_k)$
and $\|\vr^*_k\| \leq \|\vs^*_k\|$ for all $k=1,\ldots, n$.
\end{lemma}

\begin{definition}[$\fbdd$-BDD]
  Given a lattice $L \subseteq \Z^\dimin$ and a vector $\vt$ such
  that $\dist(\vt,L) <\fbdd 
  \lambda_1$, with $\fbdd \leq 1/2$, output the closest vector.
\end{definition}

The nearest plane algorithm due to Babai is an algorithm that given
$L\subseteq \Z^n$, and $\vt \in \Z^n$ returns a vector $\vv\in L$
such that $\| \vv - \vt\| \leq \dist(\vt,L) \cdot 2^{n/2}.$
BDD can be solved with this algorithm when
$2^{n/2}\dist(\vt,L) \leq \lambda_1/2$ because the answer is
unique.  The BDD problem can be solved with an approximation
factor/time tradeoff with an approximate CVP algorithm based on the
following two theorems.

\begin{theorem} [{\cite[Theorem 8]{S94}}]\label{thm:schnorr-thm8}
  Let $\vb_1,\ldots, \vb_m \in \R^n$ be a $\beta$-reduced basis and
  let $\vx = \sum_{i=1}^m x_i \vb_i^*$.  Suppose that 
  $\|\vb_k^*\| = \max(\|\vb_{m-\beta+1}^*\|,\ldots, \|\vb_m^*\|),
  m-\beta+1 \leq k \leq m$.  Let $\vv=\sum_{i=1}^mv_i\vb_i$ be a lattice
  point such that $\sum_{j=k}^m | x_j - \sum_{i=j}^m v_i \mu_{i,j}|^2
  \|\vb_i^*\|^2$ is minimal for all $v_k,\ldots, v_m\in \Z$, and
  $|x_j-\sum_{i=j}^mv_i\mu_{i,j}|\leq 1/2$ for $j=k-1,\ldots, 1$, then
  $\|\vt-\vv\|^2 \leq m \gamma_\beta^{2(m-1)/(\beta-1)}\min_{\vu\in L } \|\vt-\vu\|^2$.
\end{theorem}

\begin{theorem}[$\beta^{m/\beta}$-Approximate CVP in time
  $\beta^\beta\log B$]\label{thm:schnorr}
  There is an algorithm that on input a CVP instance $(L,\vt)$ for an
  $m$-dimensional lattice and a vector $\vt$ in the span of $L$, returns a
  vector $\vv$ such that $$\|\vt-\vv\| \leq \sqrt{m}
  \beta^{(m-1)/(\beta-1)} \min_{\vu \in L} \|\vt-\vu\|.$$
  The running time is $O(nm(\beta^{O(\beta)}+m^2)\log B)$, where $B$
  is the maximal length of the given basis vectors.
\end{theorem}

\begin{proof}
To compute the approximate closest vector, following Page 516 of~\cite{S94},
 use~\cite{S87} to compute
an ``approximate'' $\beta$-reduced basis as in
Theorem~\ref{thm:schnorr-thm8} using
$O(nm(\beta^{O(\beta)}+m^2)\log B)$ steps, then use Kannan's algorithm
to compute the closest vector using enumeration, and then use
Theorem~\ref{thm:schnorr-thm8} for the bound.  By the statement on
page 511, $\gamma_\beta \leq (2/3)\beta$ for 
$m \geq 2$.
\end{proof}

\subsection{Quantum computation}

For a positive integer $q$, let $F_q$ denote the Fourier transform over $\Z_q$. 
On a basis state with $0\leq x < q$, this operation maps $\ket{x} \mapsto \frac{1}{\sqrt{q}}
\sum_{i=0}^{q-1} \om_q^{ix} \ket{i}$ and can be computed in time $\poly(\log(q))$.  The Fourier transform over a
direct product $\Z_q \times \Z_r$ is $F_q \otimes F_r = (F_q \otimes
I)(I\otimes F_r)$, and can be 
computed on one register at a time.  

\begin{claim}
For quantum states $\ket{\phi}$ and $\ket{\psi}$,  $\|\ket{\phi}-\ket{\psi}\| = \sqrt{2(1-Re(\<\phi|\psi\>))}.$
\end{claim}

\begin{lemma}[{\cite[Lemma 3.2.6]{BV97}}] \label{lem:bv}
  For quantum states $\ket{\phi}$ and $\ket{\phi'}$, if $\|
  \ket{\phi} - \ket{\phi'} \| \leq \eps$, then the total variation 
  distance between the probability distributions resulting from 
  measurements of the two states is at most $4\eps$. 
\end{lemma}

For superpositions the representatives
$\Z_q = \{0,\ldots, q-1\}$ will be used.  
It is convenient
because of the typical quantum Fourier transform definition.  Note
that the norm $\|\cdot\|_q$ defined earlier is independent of the
choice of representatives.

Two main subroutines are for computing the quantum Fourier transform
and computing the phase of an eigenvalue of a unitary.  
Given a unitary $U$ and an eigenvector $\eva$ with eigenvalue
$\om^\phase \in \C^*$, and a power $\ppower$, the phase estimation
algorithm approximates the phase $\phase$ of the eigenvalue.  The
first step of the algorithm computes the Hadamard transform
on $\log m$ qubits and then computes the controlled-$k$-$U$ in
superposition, resulting in the {\em phase state}
$\frac{1}{\sqrt{\ppower}} \sum_{k=0}^{\ppower-1} \ket{k} \otimes U^k \eva =
\frac{1}{\sqrt{\ppower}} \sum_{k=0}^{\ppower-1} \om^{k\theta}
\ket{k} \otimes \eva$.  The Fourier transform over $\Z_\ppower$ is computed in the
first register and it is measured, resulting in a value $\phasesamp \in [\ppower]$,
where $\frac{\phasesamp}{\ppower}$ approximates $\phase$.

\begin{theorem}[Phase estimation] \label{thm:poweroftwo}
Let $\ket{\psi}$ denote a quantum state on $n$ qubits
and $U$  unitary on $n$ qubits for which
$\ket{\psi}$ is an eigenstate with eigenvalue $\phase$.
  Let $ \frac{\hat{x}}{2^t}$ be an integer multiple of
  $\frac{1}{2^t}$ closest to $\phase$.  The phase estimation algorithm
 returns $\hat{x}$ with probability at least $\frac{4}{\pi^2}$. 
 If $t=m+r+1$ and $\ppower =2^t$, then $\phasesamp \in [\ppower]$ is
 returned such that $\frac{\phasesamp}{\ppower}$ satisfies
 $|\frac{\phasesamp}{\ppower} - \phase|_1 \leq 
  \frac{1}{2^r}$ with probability at least $1-\frac{1}{2^m}$.
  The running time of the algorithm is $\poly(n, t)$ times the time to
  compute $U^T$.
\end{theorem}

\section{Approximate eigenvector of many operators}

In this section we define the notion of an approximate eigenvector and
show how well the phase estimation algorithm works compared to 
the exact eigenvector case. 

\begin{definition} \label{def:approx-ev}
  For a unitary $U$, an {\em $\eps_{ev}$-approximate eigenvector} is a
  vector $\eva$ with associated eigenvalue $\lambda \in \C^*$ satisfies $\| U \eva - \lambda \eva
  \| \leq \eps_{ev}$. This may be denoted $(U,\eva,\lambda,
  \eps_{ev})$, and where $U$ and $\eva$ are given as input.
\end{definition}

This notion will be used where one vector is used as an approximate
eigenvector of a set of unitaries.  From that point of view, it may be
helpful to say that $U$ approximates the unitary $V=\lambda I$ on the
subspace spanned by $\eva$ because $V\eva = \lambda I \eva = 
\lambda \eva$.  This may be denoted $(\lambda I,\eva,\lambda,0)$.
This also means that $|\lambda|=1$, as in the definition.

\begin{lemma} \label{lem:approx-ev}
  Let $U$ be a unitary and $\eva$ be an $\eps_{ev}$-approximate
  eigenvector with eigenvalue $\lambda$.  
  Then $\forall k,  \| U^k \eva -
  \lambda^k \eva \| \leq k \eps_{ev}$. 
\end{lemma}

\begin{proof}
  Proof by induction on $k$.  The base case is $k=1$ where $\|U^1 \eva
  - \lambda^1 \eva\| \leq 1\cdot \eps_{ev}$ by assumption.
  Assume the claim is true for $k-1$, that is,
$\| U^{k-1} \eva -
  \lambda^{k-1} \eva \| \leq (k-1) \eps_{ev}$.  Let 
  $U^{k-1} \eva = \lambda^{k-1} \eva + \eps_{k-1}\ket{E_{k-1}}$, where
    $\eps_{k-1} \leq (k-1) \eps_{ev}$.
  Then 
  \begin{align}
     \| U^k \eva - \lambda^k \eva \|
    &=\ \| U \cdot U^{k-1} \eva - \lambda^{k}\eva \|
    \\
    &=\  \| U (\lambda^{k-1} \eva + \eps_{k-1}\ket{E_{k-1}}) -
       \lambda^{k}\eva\|
    &      \text{induction hypothesis}
    \\
    &=\  \| U \eva+ \lambda^{-(k-1)}\eps_{k-1}U\ket{E_{k-1}} - \lambda 
      \eva \|
      & \text{multiply by $\lambda^{-(k-1)}$}
    \\
    &\leq \  \| U \eva- \lambda 
     \eva  \| +\|\lambda^{-(k-1)} \eps_{k-1} U \ket{E_{k-1}}\|
    \\
    &\leq \  \eps_{ev} +\eps_{k-1}\leq k\eps_{ev}.
    \end{align}
    The fact that $\lambda$ has norm one was used, and the last
    inequality is by definition of $\eps_{ev}$-approximate
    eigenvector.
\end{proof}

\begin{lemma} \label{lem:approx-ev-tvd}
  Consider the step before measuring 
  in the phase estimation routine
  run up to power $\ppower$ on an approximate eigenvector instance
  $(U,\eva,\lambda, \eps_{ev})$ versus an exact instance $(\lambda I,\eva,\lambda,0)$.
  Then the distance between these two states is at most $\ppower \eps_{ev}$. 
\end{lemma}

\begin{proof}
  If $\eva$ is an eigenvector with eigenvalue $\lambda$ for some
  operator $V$, the first step of the eigenvalue estimation algorithm when
  given power $\ppower$ would be to
  create the phase state  
  $\frac{1}{\sqrt{\ppower}}
  \sum_{i=0}^{\ppower-1}\ket{k} \otimes (V^k\eva)
  =\frac{1}{\sqrt{\ppower}}
  \sum_{i=0}^{\ppower-1} \lambda^k \ket{k} \otimes \eva$.  
  Instead, the approximate
  eigenvector instance $(U,\eva,\lambda,\eps_{ev})$ is given, and the
  state computed is $\frac{1}{\sqrt{\ppower}}
  \sum_{i=0}^{\ppower-1} \ket{k}\otimes (U^k\eva)$.  
  
By Lemma \ref{lem:approx-ev}, for each $k$, let the difference vector
be $\eps_k \ket{E_k}  = U^k\eva - \lambda^k \eva$ with $|\eps_k|
\leq k\eps_{ev}.$ Comparing the distance between the approximate eigenvector and the
exact eigenvector state before measurement gives
\begin{align}
  &\Big\| (F_\ppower\otimes I)\smfrac{1}{\sqrt{\ppower}}
    \sum_{k=0}^{\ppower-1} \ket{k} \otimes(U^k \ket{\psi})
 - (F_\ppower\otimes I)\smfrac{1}{\sqrt{\ppower}}
    \sum_{k=0}^{\ppower-1} \ket{k}\otimes (\lambda^k\eva) \Big\| \\
  &=
  \Big\| \smfrac{1}{\sqrt{\ppower}} \sum_{k=0}^{\ppower-1}
    \ket{k}\otimes (U^k \ket{\psi})
 - \smfrac{1}{\sqrt{\ppower}}
    \sum_{k=0}^{\ppower-1} \ket{k}\otimes (\lambda^k\eva) \Big\|
      \quad F_\ppower \text{ is unitary}\\
  & = 
  \Big\| \smfrac{1}{\sqrt{\ppower}} \sum_{k=0}^{\ppower-1} \ket{k}\otimes
    \left( U^k \ket{\psi} - \lambda^k\eva \right) \Big\| 
    =
    \Big\| \smfrac{1}{\sqrt{\ppower}} \sum_{k=0}^{\ppower-1} \ket{k}\otimes 
    (\eps_k \ket{E_k} )\Big\|
  \\
  & =
    \Big\| \smfrac{1}{\sqrt{\ppower}} \sum_{k=0}^{\ppower-1} \eps_k\ket{k}\otimes 
    \ket{E_k} \Big\|
    =\sqrt{ \smfrac{1}{\ppower}\sum_{k=0}^{\ppower-1} 
    \eps_k^2 }
    \quad \text{ the $\ket{k}\ket{E_k}$ states are orthonormal}
  \\
  & \leq
    \sqrt{ \smfrac{1}{\ppower}\sum_{k=0}^{\ppower-1} 
    (k \eps_{ev})^2 }
  \leq 
    \eps_{ev}\sqrt{ \smfrac{1}{\ppower} \sum_{k=0}^{\ppower-1} k^2}
  \\
  & =
    \eps_{ev}\sqrt{ \smfrac{1}{\ppower} \smfrac{1}{6} (\ppower-1)\ppower(2(\ppower-1)+1)}
    \leq
    \eps_{ev}\sqrt{ \smfrac{1}{6} (2\ppower^2 -3\ppower +1) }  \leq   \eps_{ev}\ppower/ \sqrt{3}.
\end{align}
\end{proof}


\begin{lemma}[Phase Estimation on an approximate
  eigenvector]
  \label{lem:pe-approx}
  \ \\
  There exists a quantum algorithm that on input
  $$
  {\cal T}_{in} = \left( U, \ket{\psi}, q, \eps_{ev}, p_{err} \right)
  $$
  where
   $p_{err}^2/(128\eps_{ev}) \geq 1$, $\text{etime}(U,T)$ is the time it takes
   to compute $U^T$, 
  $\ket{\psi}$ is supported on vectors in $\Z_q^n$,
  where
  $(U,\eva,\om_q^s,\eps_{ev})$ 
  is an approximate eigenvector for some $s\in \Z_q$,
  returns ${\cal O} \in \Z_q$ such that
  $$
  \Pr
  \left(\big| {\cal O} - s|_q 
  \leq 
  129q \eps_{ev} / p_{err}^2
  \right)
  \geq
  1- p_{err}.
  $$
  The running time of the algorithm is $\poly(n \log(q),
  \text{etime}(U,p_{err}/(8\eps_{ev}))$.

\end{lemma}

\begin{proof}
  Let 
  $$
  b = \lceil \log(2/p_{err}) \rceil
  \quad
  a = \lceil \log(p_{err}^{2}/ (2^7 \eps_{ev})\rceil,
  \quad
  T = 2^{a + b + 1},
  $$
  and then
  $T  \leq 2^{(\log(2/p_{err}) + 1) + (\log(p_{err}^2/(2^7\eps_{ev}) +
    1)) + 1} = \frac{2}{p_{err}} \frac{p_{err}^2}{2^7 \eps_{ev}}2^3= p_{err}/ (2^3\eps_{ev})$. 
  First consider running phase estimation on $V=\om_q^{s}I$, eigenvector $\eva$ and power 
  $\ppower$
  returns $\phasesamp$ such that
  $$
  \Pr \left(
  \left| 
   \smfrac{\phasesamp}{\ppower}-\smfrac{s}{q}\right|_1
  \leq \smfrac{1}{2^a} \right)
  \geq 1-\smfrac{1}{2^b}
  $$
  by Theorem~\ref{thm:poweroftwo}.
  By the choice of $a$, $\frac{1}{2^a}
  = 2^{- \lceil \log ( p_{err}^2/(2^7 \eps_{ev} ) ) \rceil} \leq
  2^7\eps_{ev}/p_{err}^2$, and by choice of $b$, $\frac{1}{2^b}
  = 2^{- \lceil \log(2/p_{err}) \rceil } \leq p_{err}/2$.
  Therefore
    $$
  \Pr
  \left( \left| \smfrac{\phasesamp}{\ppower}-\smfrac{s}{q} \right|_1
  \leq 
  2^7 \eps_{ev}/p_{err}^2 \right)
  \geq
  1- p_{err}/2.
  $$ 
  Scaling by $q$, the condition is equivalent to $\left| q \smfrac{\phasesamp}{\ppower}-s  \right|_q
  \leq 2^7 q \eps_{ev}/p_{err}^2 $.
  Let $ {\cal O} = \lfloor q \smfrac{\phasesamp}{\ppower}
  \rceil.$  Then
  $$
  \left| {\cal O} - s \right|_q
  = \left| \lfloor q \smfrac{\phasesamp}{\ppower}\rceil -q
    \smfrac{\phasesamp}{\ppower} + 
  q \smfrac{\phasesamp}{\ppower}-s \right|_q
  \leq  \left| \lfloor q \smfrac{\phasesamp}{\ppower}\rceil -q
    \smfrac{\phasesamp}{\ppower} \right|_q + 
  \left| q \smfrac{\phasesamp}{\ppower}-s \right|_q \leq
  \smfrac{1}{2} +  2^7 q \eps_{ev} /p_{err}^2
  \leq 129 q \eps_{ev} / p_{err}^2.
  $$
  For the error bound, consider using $U$ instead on the
  $\eps_{ev}$-approximate eigenvector $\eva$. 
  By Lemma~\ref{lem:approx-ev-tvd} and Lemma~\ref{lem:bv} the error
  increases by at most 
  $4T\eps_{ev}
    \leq  p_{err}/2.$ 
  The union bound  on the phase estimation error and the approximate
  eigenvector error gives a total error at most $p_{err}/2 + p_{err}/2\leq p_{err}$.
\end{proof}

\section{An approximate eigenvector of shift operators close to group
  elements}  

In this section a quantum state with a random phase is defined that is an
approximate eigenvector of shift operators whose shifts are close to
points in the lattice $L$.  As described in
Section~\ref{sec:lat-group}, formally the setting will be in the finite abelian group $G = L \bmod q$, which is a subgroup of $\zqn$, together
with a distance $\|\cdot\|_q$ on $\zqn$.  This setup makes it possible
to take a $q$-periodic lattice and target vector, reduce them mod $q$,
define and solve BDD over $\zqn$, and to map the solution back to the
integers.  For clarity this section will be restricted to finite groups.

\subsection{Phased Cube States and BDD on Subgroups of $\zqn$}

The approximate eigenvector is a superposition of lattice points
with a phased cube around each point.  The cube's side length controls
how much cubes around two nearby $\zqn$ points overlap.

For $\rl \in {\mathbb N}$ define the zero-centered set of ``radius''
$\rl$ as $\pmr = \{q-\rl+1,\ldots, \rl\} \subseteq \Z_q$.  The
set $\pmr$ has $\sl$ elements.

\begin{definition}\label{def:pcs}
  Let $\sl\in {\mathbb N}$ be a side length.
  \begin{enumerate}
    \item Define the {\em cube state}
 around a point $\vy \in \zqn$ by
 \[\cube{\vy} =\smfrac{1}{(\sl)^{\dimin/2}} 
   \sum_{\vz \in \pmr^\dimin} \ket{\vy +\vz }.\]
\item Let $\G$ be a subgroup of $\zqn$ with $\grank = \fgr(G)$,
  generator matrix $\matG$ 
  and coefficient space $\coeffs$.
 Define the {\em phased cube state} with label $\va \in \sampsp$ to be
 \[
   \pcs{\va}
 = \smfrac{1}{\sqrt{|G|}} \sum_{\vc \in \coeffs} \charval{\va}{\vc}
 \cube{\matG\vc }
 = \smfrac{1}{\sqrt{|G|(\sl)^\dimin}} 
 \sum_{\vc \in \coeffs} \charval{\va}{\vc}
\sum_{\vz\in \pmr^\dimin} \ket{\matG\vc +\vz}.
\]
\end{enumerate}
\end{definition}

\begin{lemma}[Cube state properties] \ \label{lem:cs-props}
  \begin{enumerate}
    \item \label{lem:cs-props-compute-cube} Given $\vy \in \zqn$, $\cube{\vy}\ket{\vy}$ is computable in time
      $\poly(\dimin \log q)$.

    \item Define the {\em shift operator}
      $U_\vx$ by
      $U_{\vx} \ket{\vy} = \ket{\vx+\vy}$.\\
      $\forall \vx,\vy \in \zqn, U_\vx \cube{\vy} = \cubep{\vx+ \vy}$,
      and the transformation $\cube{\vy}\ket{\vx}$ to
      $\cubep{\vx+\vy}\ket{\vx}$ is computable in time $poly(\dimin \log q)$.

  \item Let $\cube{\vy}$ be a cube state of side length $\sl$ and
    let $\vD\in \zqn$.
    \begin{enumerate}

    \item \label{lem:cs-props-close}  Then $\big\| \cube{\vy} -
        \cubep{\vy+\vD} \big\|
        \leq \sqrt{ \dimin \frac{\|\vD\|_q}{\rl}}$.
        
    \item \label{lem:cs-props-far} If $\|\vD\|_q \geq
      \sqrt{\dimin}\sl+1$ then
      $\cubet{\vy}  \cubep{\vy+\vD}  =0$.
    \end{enumerate}
      
 \end{enumerate}
\end{lemma}

\begin{proof}
  \begin{align}
  \ket{\vzero,\vy}
  \stackrel{F_\sl^{\otimes \dimin}}{\longrightarrow}
  &\ \smfrac{1}{(\sl)^{\dimin/2}}
  \sum_{\vz \in \Z_{\sl}^\dimin}
    \ket{\vz,\vy}
    \\
    \stackrel{(-(\rl)_i)_1}{\longrightarrow }
  &\ \smfrac{1}{(\sl)^{\dimin/2}}
  \sum_{\vz \in \pmr^\dimin}
    \ket{\vz,\vy}
    & \vz \rightarrow \vz-(\rl,\ldots,\rl), \text{ and  reindex}
    \\
    \stackrel{(+\vy)_1}{\longrightarrow }
  &\ \smfrac{1}{(\sl)^{\dimin/2}}
  \sum_{\vz \in \pmr^\dimin}
    \ket{\vy+\vz,\vy}
    \\
    =
  &\ \cube{\vy}\ket{\vy}.
  \end{align}
  
For $\vx \in \zqn$, $U_\vx \cube{\vy} = \frac{1}{(\sl)^{\dimin/2}} \sum_{\vz \in \pmr^{\dimin}}
U_\vx \ket{\vy+\vz }
= \frac{1}{(\sl)^{\dimin/2}}
\sum_{\vz \in \pmr^{\dimin}} \ket{\vx +\vy+\vz}
= \cubep{\vx+ \vy}$.  Therefore, given $\cube{\vy}\ket{\vx}$, one
addition from the
second register into the first register results in $\cubep{\vx+\vy}\ket{\vx}$.

For (\ref{lem:cs-props-close}), start with
$$
\Big\| 
      \cube{\vy}-\cubep{\vy + \vD} \Big\|^2 =
2 \cdot (1 - \Re(\cubet{\vy}\cubep{\vy + \vD} )).
$$
Since
\begin{align}
& \cubet{\vy} \cubep{\vy+\vD } 
= \cubet{\vzero}\cube{\vD}
        &\text{invariant under shift by $-\vy$}
  \\ & = \ 
       \frac{1}{(\sl)^\dimin} \sum_{\vz \in \pmr^\dimin} \sum_{\vz' \in \pmr^\dimin}
  \<\vz|\vz'+ \vD\>
\\ &  =\ \frac{\text{\# common points in } \pmr^\dimin \text{ and } \pmr^\dimin+\vD }{(\sl)^\dimin}
  \\ & \geq \ 
       \frac{(\sl - \|\vD\|_q)^{\dimin}}{(\sl)^{\dimin}}
    = 
       (1-\smfrac{\|\vD\|_q}{2\rl})^\dimin
\geq       1-\dimin \smfrac{\|\vD\|_q}{2\rl}.
       &\text{Bernouli's inequality}
\end{align}
Therefore 
$\big\| \cube{\vy} - \cubep{\vy+\vD} \big\|
\leq \sqrt{\dimin\frac{\|\vD\|_q}{\rl}}$.

For (\ref{lem:cs-props-far}) assume $\|\vD\|_q \geq \sqrt{\dimin}\sl+1$.  To have a common point, there must 
exist $\vz, \vz' \in \pmr^\dimin$ such that $\vz =
\vz'+\vD \in \zqn$.  For this to happen, $\vz - \vz' =
\vD$, and so $ \sqrt{\dimin}\sl+1 \leq \| \vD\|_q  = \| \vz -
\vz'\|_q \leq \| \vz\|_q + \|\vz'\|_q \leq
2\sqrt{\dimin}\rl$, which is a contradiction, so no points 
are in common.
\end{proof}

\begin{algorithm}[Computing a PCS state] \label{alg:pcs-state}
\ \\
  Input:  A decomposed subgroup $\groupdecomp$ of $\zqn$ 
  and a cube side length $\sl$.\\
  Output: $\pcs{\va}$ and $\va \in_R \sampsp$.\\
  \begin{align}
    \ket{\vzero,\vzero} 
    \longrightarrow
    & \ \cube{\vzero}\ket{\vzero}
    &\text{ by Lemma~\ref{lem:cs-props}(\ref{lem:cs-props-compute-cube}}) 
    \\
    \stackrel{ (F_\coeffs)_2}{\longrightarrow}
    &\ \smfrac{1}{\sqrt{|G|}} 
      \sum_{\vc \in \coeffs}
      \cube{\vzero}\ket{\vc}
    \\ 
  \stackrel{(+\matG\vc)_1}{\longrightarrow}
  &\ 
    \smfrac{1}{\sqrt{|G|}} 
    \sum_{\vc \in \coeffs}
    \cube{\matG\vc}\ket{\vc}
    \\ 
  \stackrel{(F_{\sampsp})_2}{\longrightarrow}
  &\ 
    \smfrac{1}{\sqrt{|G|}}
    \sum_{\vc \in \coeffs}
    \cube{\matG\vc}
    \Big(
    \smfrac{1}{\sqrt{|\sampsp|}}
    \sum_{\va \in \sampsp}
    \charval{\va}{\vc}
    \ket{\va} \Big)
    \\
  \stackrel{M_2}{\longrightarrow}
  &\ 
    \smfrac{1}{\sqrt{|G|}} 
    \sum_{\vc  \in \coeffs}
    \charval{\va}{\vc}
    \cube{\matG\vc } \ket{\va}
    \\
    =& \ 
    \pcs{\va}\ket{\va}.
\end{align}
\end{algorithm}

\begin{lemma}[Phased Cube State Approximate Eigenvector Properties] \
  \\ \label{lem:pcs-props}
  Let $\G$ be a subgroup of $\zqn$ with decomposition $\groupdecomp
  \in \Z_q^{\dimin \times \grank}\times \Z^\grank$, with
  shortest (nonzero) element length 
  $\lambda_1 = 
  \min_{\vv \in G\backslash\{ \vzero \}} \|\vv\|_q$, 
  and a phased cube state $\pcs{\va}$ for $\va \in \sampsp$ and side
  length $\frac{1}{4} \frac{\lambda_1}{\sqrt{\dimin}} \leq \sl \leq \frac{1}{2}\frac{\lambda_1}{\sqrt{\dimin}}$.  Then
  \begin{enumerate}
  \item \label{pcs-props-ev}  If $\vc\in \coeffs,$ let $\matG\vc 
    = \vv \in G$, and then $U_{\vv} \pcs{\va} 
    = \charval{\va}{-\vc} \pcs{\va}$.

  \item \label{pcs-props-approx} $\forall \vD \in \zqn$ with $\| \vD \|_q \leq \frac{\lambda_1}{2},$ 
    $\|U_{\vD } \pcs{\va} - \pcs{\va}\|
    \leq 4 \dimin^{3/4} \sqrt{\smfrac{\| \vD \|_q}{\lambda_1}}.$

  \item \label{pcs-approx-ev} $\forall \vy \in \zqn$, let
    $\matG\vs \in G$ with coefficients $\vs \in \coeffs$ be such that
    $\vD= \vy-\matG\vs$ satisfies
    $\|\vD\|_q \leq \frac{\lambda_1}{2}$.  Then
    $\|U_{\vy } \pcs{\va} - \charval{\va}{-\vs} \pcs{\va}\| \leq 4
    \dimin^{3/4} \sqrt{\smfrac{\| \vD\|_q}{\lambda_1}} =:
    \eps_{ev}.$ In particular, for all $\va \in \sampsp$, the
    state $\pcs{\va}$ is an $\eps_{ev}$-approximate eigenvector of
    $U_{\vy}$ for any element $\vy$ where $\dist_q(\vy,\G) \leq \frac{\lambda_1}{2}$. 

  \item \label{pcs-compute} Given a decomposed group $\groupdecomp$, Algorithm
    \ref{alg:pcs-state} with side length $\sl$ computes the state $\pcs{\va}\ket{\va}$
    in time $\poly(\dimin \log q)$, where $\va$ is a uniformly chosen
    element from $\sampsp$.

  \end{enumerate}
\end{lemma}

\begin{proof}
For (\ref{pcs-props-ev}), let $\vv= \matG\vc \in G$ and $\va \in \sampsp$, then
\begin{align}
  U_{\vv} \pcs{\va} 
  &= \smfrac{1}{\sqrt{|G|}}
        \sum_{\vd\in \coeffs} \charval{\va}{\vd}
    U_{\vv} \cube{ \matG\vd } \\
  & = \smfrac{1}{\sqrt{|G|}}
        \sum_{\vd \in \coeffs} \charval{\va}{\vd}
    \cubep{\matG\vc+\matG\vd } \\
  &= \smfrac{1}{\sqrt{|G|}}  
        \sum_{\vd' \in \coeffs}
    \charval{\va}{\vd'-\vc} \cube{\matG\vd'} 
    &\text{ reindex coeffs with } \vd'=\vc+\vd \\
  &= \smfrac{1}{\sqrt{|G|}}  
        \sum_{\vd' \in \coeffs}
    \charval{\va}{\vd'} \charval{\va}{-\vc} \cube{\matG\vd'} 
\\
    &= \charval{\va}{-\vc} \pcs{\va}.
\end{align}

For (\ref{pcs-props-approx}), let $\vD\in \zqn$, and then
\begin{align}
  &\Big\|U_{\vD } \pcs{\va} - \pcs{\va} \Big\|
    \\ &
  = \ \Big\| U_\vD \smfrac{1}{\sqrt{|G|}} \sum_{\vc \in \coeffs}
         \charval{\va}{\vc} \cube{\matG\vc} -
         \smfrac{1}{\sqrt{|G|}} \sum_{\vc \in \coeffs} \charval{\va}{\vc} \cube{\matG\vc} \Big\|
  \\ &
=\  \Big\| \smfrac{1}{\sqrt{|G|}} \sum_{\vc \in \coeffs}
     \charval{\va}{\vc} \cubep{\matG\vc+\vD} -
         \smfrac{1}{\sqrt{|G|}} \sum_{\vc \in \coeffs}
       \charval{\va}{\vc} \cube{\matG\vc} \Big\| 
  &\text{apply $U_\vD$}
  \\ & 
  = \ \Big\| \smfrac{1}{\sqrt{|G|}} \sum_{\vc \in \coeffs} 
      \charval{\va}{\vc} \left( \cubep{ \matG\vc + \vD} -
     \cube{\matG\vc} \right)\Big\|
      &\text{group terms}
  \\&  
   = \ \Big\| \smfrac{1}{\sqrt{|G|}} \sum_{\vc \in \coeffs} 
  \charval{\va}{\vc} \eps_{\vc} \ket{E_{\vc}} \Big\|
   \text{ by Lemma \ref{lem:cs-props}(\ref{lem:cs-props-close})},
\\
  & \quad \quad \quad \quad
     \quad \quad \text{ let } \eps_\vc \ket{E_\vc}
    \text{ be the     difference,}
  \text{ with }
      |\eps_\vc | \leq \sqrt{ \dimin \smfrac{\| \vD\|_q}{\rl}} 
  \\ &
  \leq \ \max_{\vc \in \coeffs} |\eps_\vc|
       \leq \sqrt{ \dimin \smfrac{\|\vD\|_q}{\rl}}
        \leq \ \sqrt{8 \dimin^{3/2} \smfrac{\| \vD \|_q}{\lambda_1}}
        \leq
        4 \dimin^{3/4} \sqrt{\smfrac{\| \vD \|_q}{\lambda_1}}.
\end{align}

For (\ref{pcs-approx-ev}), let $\vy \in \zqn$ satisfy $\dist_q(\vy, G) \leq
    \eps_1 \lambda_1$, let $\matG\vs \in G$ be the closest element to
    $\vy$, and let $\vD = \vy  - \matG\vs$.  Then
    $\|U_\vy  \pcs{\va} - \charval{\va}{-\vs} \pcs{\va} \| 
    =\|U_{\matG \vs +\vD} \pcs{\va} - \charval{\va}{-\vs}
    \pcs{\va}\|
    =|\charval{\va}{-\vs}| \cdot \|U_{\vD } \pcs{\va} - 
    \pcs{\va}\|
    \leq 4 \dimin^{3/4} \sqrt{\smfrac{\| \vD \|_q}{\lambda_1}}$
    by part (\ref{pcs-props-ev}) and then part (\ref{pcs-props-approx}).

Finally, for (\ref{pcs-compute}), the algorithm computes the Fourier
transform, addition in $\zqn$, and measurement, which are all
polynomial time in $n$ and $\log q$.  
By Lemma \ref{lem:cs-props} the $\cube{\vv}$'s form an orthogonal set of states.  
The probability of measuring $\va$ is hence proportional to
$\left\| \sum_{\vc \in \coeffs} \charval{\vc}{\va}  \cube{\matG\vc} \right\|^2
=
\sum_{\vv \in G}
\left\| \cube{\vv} \right\|^2$
which is independent of $\va$.
\end{proof}

Using PCS states and the phase estimation algorithm on approximate eigenvectors
we define the following algorithm which outputs LWE samples but
with error different than the typical Gaussian error.  In order keep
the terminology more closely related to lattices we phrase it as
sampling inner products.
\begin{mdframed}
\begin{algorithm}\label{alg:random-hip}

  \textbf{${\rm SampleHIP}(\groupdecomp, \sl, \eps_1,
    \vt, p_{err})$}
  \ \\
  \noindent Input: finite group decomposition $\groupdecomp$, side
  length $\sl$, target vector $\vt\in \Z_q^n$. 

\begin{enumerate}
\item
  Compute a PCS state using Algorithm \ref{alg:pcs-state} on
  $\groupdecomp$ and $\sl$ to get a random label $\va$ and state
  $\pcs{\va}$ on $(\dimin +\grank)\log q$ qubits.

\item
Let $\eps_{ev} =  \sqrt{\eps_1}  \cdot 4 \dimin^{3/4}$.
 Run phase estimation on
$
  {\cal T}_{in} = 
  \left( U_\vt, \pcs{\va}, q, \eps_{ev},  p_{err} \right) 
 $ to get output $\cal O$
 and return $(\va,{\cal O}) \in \Z_{q}^\grank \times\Z_{q}$.

\end{enumerate}

\end{algorithm}
\end{mdframed}

\begin{lemma}[Sampling Hidden Inner Products]
  \label{lem:sample-g} 
  
\noindent
If $\sl \in \left[
   \frac{1}{4}\frac{\lambda_1(\G)}{\sqrt{\dimin}},
   \frac{1}{2}\frac{\lambda_1(\G)}{\sqrt{\dimin}}\right]$,  $0 < \eps_1 < \frac{1}{2}$
 and $\dist_q(\vt,\G) \leq \eps_1\lambda_1$ 
 then Algorithm \ref{alg:random-hip} 
 runs in time $\poly(\dimin \log q,
 \log(p_{err}/\eps_{1}))$ and
 returns a uniformly
 random $\va \in \sampsp$, and an ${\cal O}\in \Z_q$ satisfying 
 $$
 \Pr \left( \left| {\cal O} - \charphase{\va}{(-\vs)} \right|_q 
 \leq
129 \cdot q \cdot \sqrt{ \eps_1} \cdot 4 \dimin^{3/4}  \cdot p_{err}^{-2}
\right) \geq 1 -  p_{err}.
$$ 
where the element $\matG\vs$ is the closest group element to
 $\vt$ in $G$. 
\end{lemma}

\begin{proof}
In the first step of Algorithm \ref{alg:random-hip}
on input $\groupdecomp$ with given side length $\sl$
Algorithm \ref{alg:pcs-state}
returns a state $\pcs{\va}$ on $(\dimin +\grank)\log q$ qubits, where $\va \in \sampsp$
is uniformly random by Lemma \ref{lem:pcs-props}.  
Since by assumption $\dist_q(\vt,G) \leq \eps_1 \lambda_1 <
\frac{\lambda_1}{2}$ the element $\matG\vs$ is the 
closest group element to $\G$.

In the next step the
tuple  
  ${\cal T}_0 = 
  \left( U_\vt, \pcs{\va}, \charval{\va}{-\vs}, \eps_{ev}
    =  \sqrt{\eps_1}  \cdot 4 \dimin^{3/4}\right)$
  is an approximate eigenvector instance by Lemma \ref{lem:pcs-props}.
  In addition, since $U_\vt$ is the shift operator, $U_\vt^t$
  can be computed in time $\poly(n, \log(t))$ by repeated  
  squaring.
  It follows by Lemma~\ref{lem:pe-approx} that running phase
  estimation on the tuple 
  ${\cal T}_{in} = 
  \left( U_\vt, \pcs{\va}, q, \eps_{ev},  p_{err} \right)$
  returns a value ${\cal O}\in \Z_q$  that satisfies
  $$
  \Pr
  \left(
  | {\cal O} - \charphase{\va}{(-\vs)} |_q 
  \leq  
  129 \cdot q \cdot  \eps_{ev} \cdot p_{err}^{-2} 
  =
  129 \cdot q \cdot \sqrt{\eps_1} \cdot 4 \dimin^{3/4} \cdot p_{err}^{-2}
  \right)
  \geq
  1- p_{err}.
  $$
 The running time requires $\poly(\dimin \log
 q)$ time for the PCS generation Algorithm \ref{alg:pcs-state} and
 an additional $\poly(\dimin \log q, \log(p_{err}/\eps_{ev}))$
 for the phase estimation algorithm on approximate eigenvectors by
 Lemma \ref{lem:pe-approx}, for a total of
 $ \poly(\dimin \log q, \log(p_{err}/\eps_1)). $
\end{proof}

\section{Random Self Reducibility for subgroups $G$ of $\zqn$ of finite group rank $\grank$}
\label{sec:rsr}

This section establishes RSR for instances of $\bdd$ over subgroups of $\zqn$ using a quantum algorithm. 
\begin{mdframed}
\begin{algorithm}{${\rm SampleBDD}(\groupdecomp, \hat{\lambda}_1, \vt, \eps_1, \dimout, p_{err})$}\label{alg:samplelattice2}

\begin{enumerate}

  \noindent
Input: A decomposed subgroup $ \groupdecomp$ of $\zqn$, a target
vector $\vt\in \zqn$, 
an estimate of the length of the shortest vector $\hat{\lambda}_1$,
required error probability $p_{err}$,
and a target dimension $\dimout$.

\item
Let $p_{err}^{PE} =  p_{err}/(2\dimout)$, $\sl = \frac{1}{2}\frac{\hat{\lambda}_1}{\sqrt{\dimin}}$.

\item 
For each $i\in [\dimout]$ run
$$
\left( \va_i \in \Z_q^\grank, {\cal O}_i\in \Z_q \right) 
=
\samplehip(\groupdecomp, \sl, p_{err}^{PE}, \eps_1, \vt).
$$

\item
Let $\rndmatG\in \zq^{\dimout\times \grank}$ have rows formed by the
$\va_i$, let $\tilde{\vg}_i$ denote the $i$th column of $\rndmatG$, and let
$\rndG = \langle \tilde{\vg}_1, \ldots,  
\tilde{\vg}_\grank  \rangle \subseteq \Z_q^\dimout$. 
\item
Define the target element $\tilde{\vt} = ( {\cal O}_1, \hdots, {\cal O}_\dimout)$.

\item
Return $\rndmatG, \tilde{\vt}$.

\end{enumerate}
\end{algorithm}
\end{mdframed}

\begin{lemma}[Coefficient-preserving random group sampling]\label{lem:samplebdd}

\noindent
Let $\groupdecomp$ be a decomposed subgroup
of $\zqn$, $\vt\in \zqn$, and $\eps_1>0$.  
Let $\hat{\lambda}_1$ be an estimate of $\lambda_1$ such that $\hat{\lambda}_1 \in [\lambda_1, 2 \lambda_1]$.
If $\dist_q(\vt, \G) \leq \eps_1 \lambda_1(\G)$,
then 
${\rm SampleBDD}(\groupdecomp, \hat{\lambda}_1, \vt, \eps_1, \dimout, p_{err})$
returns a random subgroup $\rndG = \langle \tilde{\vg}_1,
\ldots, \tilde{\vg}_\grank \rangle \subseteq \zq^{\dimout}$, a vector
$\tilde{\vt} \in \Z_q^{\dimout}$ such
that w.p.\ at least $1 - p_{err}$, when $p_{err}/2 \geq 1/2^m + 1/q^{m-\grank}$,
\begin{enumerate}
\item Preservation of bounded distance to lattice:\\
$
\dist_q( \tilde{\vt}, \rndG ) 
\leq 
\sqrt{\eps_1}  \cdot 
q^{\grank/\dimout} \lambda_1(\rndG) \cdot 260   \dimin^{3/4}\dimout^{2.5} p_{err}^{-2}  
$.

\item Preservation of coefficients: \\
\noindent
If $\vs\in \Z_q^\grank$ is such that $\|\vt - \matG \vs \|_q \leq
\eps_1\lambda_1(\G)$, 
then $\| \tilde{\vt} - \rndmatG \vs \|_q \leq \eps_1\lambda_1(\rndG)$,

Let $\vs$ be such that $\matG \vs \in \G$,  denote the closest
group element to
$\vt$ in $\G$.
Then 
$\vs$ 
is
the unique element such that $\rndmatG \vs$ is the closest element to 
$\tilde{\vt}$ in $\rndG$.

\end{enumerate}
The running time of the procedure is
$\poly(\dimin \log q,\dimout, \log(\dimout \cdot p_{err}/\eps_1))$.
\end{lemma}

\begin{proof}

For the distance, by 
Lemma~\ref{lem:sample-g} and the union bound over $\dimout$ samples, 
$$
\Pr
\left(
\forall i\in [\dimout],
\left| {\cal O}_i -  \charphase{\va_i}{(-\vs)} \right|_q  \leq 
129 \cdot q  \sqrt{\eps_1}
4 \dimin^{3/4}  (p_{err}^{PE})^{-2}  
\right)
\geq
1-\dimout \cdot p_{err}^{PE} \geq  1- p_{err}/2.
$$
When this condition holds,
\begin{align}
  \dist_q( \tilde{\vt}, \rndG ) 
  \leq &\  \|\tilde{\vt} - \rndmatG \vs \|_q   \\
  = &\ 
      \|  ( {\cal O}_1, \hdots, {\cal O}_\dimout )
      -  (\charphase{\va_1}{\vs}, \ldots, 
      \charphase{\va_{\dimout} }{\vs} ) \|_q \\
  \leq &\ 
         \sqrt{\dimout} \max_i \Big| {\cal O}_i
         -   \charphase{\va_i}{\vs} \Big|_q \\
  \leq &\ 
         129 \sqrt{\dimout} q \sqrt{\eps_1} 4 \dimin^{3/4} (p_{err}^{PE})^{-2} 
   &      \mbox{ by Lemma \ref{lem:sample-g}}
         \\
  \leq &\
         \delta 129 \cdot \sqrt{\eps_1} \rndlambda_1(\rndG) q^{\grank/\dimout} 4 \dimin^{3/4}  \cdot \sqrt{\dimout} \cdot(p_{err}/(2\dimout))^{-2}
         &\mbox{ w.p.\ $\geq 1-1/2^\dimout$ by Claim \ref{fact:randomq}}
          \\
 = &\ 
        \delta 129 \cdot \sqrt{\eps_1} \rndlambda_1(\rndG) q^{\grank/\dimout}
     4^2 \dimin^{3/4}  \dimout ^{2.5} p_{err}^{-2}.
\end{align}
Hence by the union bound, the distance is bounded w.p.\ at least $1 -
p_{err}/2 - 1/2^{m}$. 

For the secret coefficients $\vs$, let $\rndmatG$ denote the $\dimout \times \grank$ matrix whose rows are $\va_i$.
Observe that by definition we have that $\rndmatG \cdot \vs$ is at distance at most $\dist_q( \tilde{\vt}, \rndG )$ from $\tilde{\vt}$.
Hence for $\vs$ to be the unique vector that satisfies this equation we need to additionally require
that $\rndmatG$ is primitive.
By Claim \ref{fact:randomq} we have
$\Pr( \rndmatG \mbox{ is primitive}) \geq 1 - 1/q^{m-\grank}$.

It follows that both conditions of the lemma are satisfied w.p.\ at least 
$1 - p_{err}/2 - 1/2^{m} - 1/q^{\dimout-\grank} \geq 1 - p_{err}$.

By Lemma \ref{lem:sample-g} the running time is $\poly(\dimin \log q,
\log(p_{err}/\eps_1), \dimout)$.
\end{proof}

\begin{theorem}
  \label{thm:main1}
  There is a $\poly(n, \log q)$-time quantum algorithm solving
  $2^{-\Omega(\sqrt{\grank \log q})}$-BDD on lattices of dimension
  $n$, periodicity $q$, and finite  group rank $\grank$.
\end{theorem}

\begin{proof}
  Compute the finite abelian group decomposition $\groupdecomp$ of $L(\matB)
  \bmod q$.   
  Call SampleBDD with $\groupdecomp$, $\hat{\lambda}_1$,
  $\vt \bmod q$, $\fbdd$, $m = \sqrt{\grank \log q}$, $p_{err}=1/10$.  By
  Lemma~\ref{lem:samplebdd}, $\rndmatG$ and $\tilde{\vt}$ are returned 
  satisfying
  \begin{align}
    2^m \dist_q(\tilde{\vt},\rndG) 
    & \leq \  2^m \sqrt{\fbdd} \rndlambda_1(\rndG) q^{\grank/m}\delta 129\cdot 
      4^2 n^{3/4}m^{2.5} 100 
    \\
    & \leq\ \sqrt{\fbdd} 2^{m+\grank\log q/m}\rndlambda_1(\rndG) \delta 129\cdot 
      4^2 n^{3/4}m^{2.5} 100 
    \\
    & = \ c \sqrt{\fbdd} 2^{2\sqrt{\grank\log q}}\rndlambda_1(\rndG) n^{c'}
    \\
    & \ \leq \rndlambda_1(\rndG)/2 
  \end{align}

  Choose $\fbdd = (c  2^{-4\sqrt{\grank\log q}}n^{-2c''})$.  Mapping
  back up to the integers, Babai's
  nearest plane algorithm computes the closest vector $\bmod q$ in
  polynomial time, and that can be mapped to the original problem.
  See the steps of the algorithm in the next section for details.
\end{proof}

As discussed in the introduction, $r \log q < n^2$,  
$2^n$-approximation algorithm.  For example, subexponential
approximation factors such as $\fbdd = 2^{-n^\delta}$ are possible when
$\sqrt{\grank \log q}  \leq n^\delta$, or $r \log q \leq n^{2\delta}$.

\section{Trading off running time for approximation factor}
\label{sec:schnorr}

The algorithm in this section applies Schnorr's hierarchy theorem as a
black box, rather than applying Babai's algorithm.  The algorithm will
first compute the finite abelian 
group decomposition of the lattice, solve BDD in that group, then map
the solution back to the integers.

\begin{mdframed}
  \begin{algorithm}[Algorithm ${\cal Q}$: $\alpha$-BDD on $\dimin$-dimensional $q$-periodic lattices]
    \ \label{alg:tradeoff}

\noindent
Input: A lattice basis $\matB \in\Z^{\dimin\times \dimin}$, for a
$q$-periodic lattice with finite group rank $\grank$, a
target vector $\vt \in \Z^\dimin$, a trade-off parameter
$\beta$.

\begin{enumerate}

\item 
Compute the finite abelian group decomposition $\groupdecomp$ of $L(\matB)
\bmod q$. 

\item Iterate over powers of 2 for $\hat \lambda_1 = 2,4,8 \ldots, q$:
 
\begin{enumerate} 
\item\label{it:samplebdd}
Let $\dimout =(\beta-1) \sqrt{\beta \grank \log q} +1\leq
\sqrt{\smfrac{\beta \grank \log q}{\log \beta}}$  and sample random instance of $\bdd$ in $\Z_q^n$:
$$
(\rndmatG, \tilde{\vt} ) = 
\samplebdd( \groupdecomp, \hat\lambda_1, \vt \bmod q, 
\eps_1,
\dimout, p_{err} = 1/10).  
$$
\item
  With operations over the integers, compute HNF$([\ \rndmatG\ |\
  q\cdot\matI\ ])$
  to get a basis $\rndmatB\in
  \Z^{m\times m}$.  Treating $\tilde{\vt}$ as an integer vector, 
  run the $\sqrt{m}\beta^{(m-1)/(\beta-1)}$-approximate CVP algorithm
  (Theorem \ref{thm:schnorr}) on $(\rndmatB,\tilde{\vt})$.
  Denote output by $\tilde{\vv} \in L(\rndmatB)$.

\item
Compute $\vs \in \Z^\grank$ such that $\tilde{\vv} = \rndmatB \vs$.  If
$\|\vt -\matB\vs\|_q \leq \hat \lambda_1/2$, then 
output $\matB \vs + \vt - (\vt \bmod q)$.
\end{enumerate}
\end{enumerate}

\end{algorithm}
\end{mdframed}

 \begin{theorem}\label{thm:main}
   Let $L(\matB)$ be an $n$-dimensional $q$-periodic lattice with
   finite group rank $\grank$.
   Given an instance of $\eps_1$-BDD $(\matB, \vt)$, with $2\leq
   \beta \leq \grank\log q$, and
  $$
  \eps_1 = \left(\exp\Big(\!- 4\sqrt{\smfrac{\grank \log q \log
        \beta}{\beta}} \Big)
    \cdot 2^2 m p_{\ref{lem:samplebdd}}(n,\log q)^2 \right).
  $$ 
  Algorithm \ref{alg:tradeoff} runs
  in quantum time $\approx \beta^\beta \poly(n, \log(q))$ and returns the closest vector to $\vt$ with probability at least $0.9$.      
\end{theorem}

\begin{proof}

For a lattice of dimension $m$, Schnorr's algorithm returns $\tilde{\vv}$ such that  
\begin{align}
  \|\tilde{\vv} -\tilde{\vt}\|
  & \leq \sqrt{m}\beta^{\frac{m-1}{\beta-1}}
    \dist(\tilde{\vt},L(\rndmatB)) 
&  \text{by Theorem \ref{thm:schnorr}} \\
  & \leq  \sqrt{\eps_1} \beta^{\frac{m-1}{\beta-1}} q^{r/m}
   \sqrt{m} p_{\ref{lem:samplebdd}}(n,\log q)\lambda_1(\rndG) 
&   \text{by Lemma \ref{lem:samplebdd}, w.p.\ .9} \\
  & = \sqrt{\eps_1}
    \exp\Big(\smfrac{(m-1)\log \beta}{\beta-1}
    +\smfrac{r \log q}{m} \Big) \sqrt{m} p_{\ref{lem:samplebdd}}(n,\log q)
    \lambda_1(\rndG)  \\ 
  & \leq \sqrt{\eps_1}
    \exp\Big(2\sqrt{\smfrac{ \grank \log q \log \beta}{\beta}}\Big)
    \sqrt{m} p_{\ref{lem:samplebdd}}(n,\log q)  \lambda_1(\rndG)
  & \text{Choose } m = \sqrt{\smfrac{\beta \grank \log q}{\log \beta}}\\
  & \leq \lambda_1(\rndG)/2.
\end{align}
The dimension $m$ was chosen to maximize the decoding radius for a given
$\grank \log q$ and $\beta$.  
\end{proof}

\begin{corollary}
  If $\gr \log q= n\log n$,  $0\leq \eps\leq 1/2$,
  then BDD can be solved with factor $\eps_1 \approx
  \exp\Big(\!\!-4 n^{\eps} \log n\Big)$
  in time $(n^{1-2\eps})^{n^{1-2\eps}}\poly(n,\log q)$
  $\leq 2^{n^{1-2\eps} \log n} \poly(n, \log q)$.
\end{corollary}

This running time has an $n^\eps$ factor less than using Schnorr
directly, without reanalyzing it for $q$-periodic lattices.

\begin{proof}
  Choose $\beta=n^{1-2\eps}$ and then
  $\exp\Big(\!\!-4\sqrt{\smfrac{\grank \log q \log 
      \beta}{\beta}}\Big) =
  \exp\Big(\!\!-4\sqrt{\smfrac{n \log n \log 
      n^{1-2\eps}}{n^{1-2\eps}}}\Big) =\exp\Big(\!\!-4 n^{\eps} \log
  n\Big)$.  The time is $(n^{1-2\eps})^{n^{1-2\eps}}\poly(n,\log q)$
  $\leq 2^{n^{1-2\eps} \log n} \poly(n, \log q)$.
\end{proof}

\begin{corollary}
If $\gr \log q= (\log n^c)^3$ and $\grank = \log n^c$, choose
$\beta=\log n^c$,
and then BDD can be solved with factor $\eps_1 \approx
\exp\Big( - 4\sqrt{\smfrac{\grank \log q \log \beta}{\beta}}\Big) 
= \exp\Big( - 4\sqrt{\frac{ (\log n^c) (\log n^c)^2 \log \log
    n^c}{\log n^c}}\Big)$ 
\\ $=  \exp\Big(\! -4{\sqrt{(\log n^c)^2 \log \log n^c}}\Big)= n^{-4c
  \sqrt{\log \log n^c}}$ 
in
time $(\log n^c)^{\log n^c}\poly(n,\log q)$\\
$\leq 2^{\log n^c \log \log n^c}\poly(n, \log q) = n^{c \log \log
  n^c}\poly(n,\log q)$
because $\beta= \log n^c$. \\
\end{corollary}

\newcommand{\VV}[2]{\|\vc^*_{#1}\|\cdots \|\vc^*_{#2}\|}

\section{Using LLL for BDD and SIVP on rectangle-periodic lattices}
\label{sec:lll}

In this section extend the algorithm in~\cite{DW21} to meet the bounds
of Theorem~\ref{thm:main1} for polynomial-time algorithms and
generalize it further to handle rectangle-periodic lattices.
We also
record that it can be used to solve SIVP for related parameters.  The
quantum algorithm also handles rectangle-periodic lattices without
change.

More details about~\cite{DW21} can also be found
in~\cite{ABCG21}, and about LLL in~\cite{MG02}.   

Call a lattice $L\subseteq \Z^n$ {\em $\vr$-periodic} for a vector
$\vr = (r_1,\ldots,r_n)$ if $H:= r_1 \Z \times \cdots \times r_n\Z$ is
a subgroup of $L$.  All information about the lattice is contained in
the finite abelian group
$L/H \leq \Z_{r_1}\times\cdots \times \Z_{r_n}$.  Let $\groupdecomp$
be the finite group decomposition of $L/H$. This notation uses $\vr$
and $r_i$ for the rectangle side lengths, and is separate from the
finite group rank $\gr$.  Recall that the coefficient space is
$\Z_{q_1} \times \cdots \times \Z_{q_r}$ with generators
$\vg_1,\ldots,\vg_r \in \Z_{r_1}\times \cdots \times \Z_{r_n}$.  Let
$q = \max \{r_1,\ldots, r_n\}$.

\begin{lemma}
  \label{lem:mingslen}
  There is an algorithm running in time $\poly(\gr, \ln q)$ that takes
  a lattice $L\subseteq \Z^n$ and returns a basis with minimum
  Gram-Schmidt length at least $\exp(-c\sqrt{\gr \ln q}) \cdot
  \lambda_1$, for a constant $c$. 
\end{lemma}

\begin{proof} 
  Given a basis $\matB'$ of a lattice, compute the minimum
  axis-aligned rectangle in $L$ computing a lattice vector on each
  axis, and dividing by the gcd of the coordinates in each axis.  Next
  compute the finite group decomposition $\groupdecomp$, and lift
  $\vg_1,\ldots,\vg_r$ to $\Z^n$.  Sort the lengths and label so that
  $r_1 \leq \ldots \leq  r_n$ and run the algorithm in Lemma~\ref{lem:MGalg}
  on a basis from 
  HNF$([\vg_1,\ldots,\vg_r, r_1\ve_1, \ldots, r_n\ve_n])$ and the linearly
  independent set $S=\{r_1\ve_1,\ldots,r_n\ve_n\}$, resulting in basis
  $\matB$.  Then run LLL on with constant $\Delta >1$ and let $C = $
  LLL($\matB$, $\Delta$).

  The Lovasz Condition (LC) in LLL implies for $j \leq i$,   $\frac{\|\vc_j^*\|}{\Delta^{i-j}} \leq  
  \cdots \leq  \frac{\|\vc^*_{i-2}\|}{\Delta^2} \leq  
  \frac{\|\vc^*_{i-1}\|}{\Delta^1}  \leq \frac{\|\vc_i^*\|}{\Delta^0}$.

  Let $m$ be the cutoff point for case 1 and case 2, to be chosen below.

  In case 1, defined by $i \leq m$,
  \[\frac{\lambda_1}{\Delta^{m-1}} \leq
    \frac{\|\vc_1^*\|}{\Delta^{m-1}} \leq \cdots \leq
    \frac{\| \vc_i^* \|}{\Delta^{m-i}} \leq \cdots
    \leq\frac{\|\vc_m^*\|}{\Delta^0},
  \]
  which by the choice of $m$ below,
  satisfies the bound.  The first inequality is because $0 \neq \vc_1^* = \vc_1
  \in L$.

  For case 2, assume $m\leq i$.

  \begin{enumerate}
  \item Using LC again, multiply $m$ consecutive vector lengths to get
    \[
      \frac{\VV{i-m+1}{i}}{\Delta^{m(m-1)/2} } =
      \frac{\VV{i-m+1}{i}}{\Delta^{\sum_{j=0}^{m-1} j} } \leq  
      \|\vc_i^*\|^m.
    \]

    This will be used below.

\item $\vr$-periodic, block of $m$ vectors.  Let $R$ contain the indices of the
  $\gr$ largest values in $\vr = (r_1,\ldots,r_n)$.  When sorted it
  will contain $r_{n-\gr+1},\dots,r_n$. Then

  \begin{align*}
    \frac{r_{i-m+1} \cdots r_{i}}{\prod_{j\in R} r_j}
    &
      = \frac{1}{(r_1 \cdots r_{i-m}) \cdot (r_{i+1} \cdots r_n)}
    \frac{r_1\cdots r_n}{\prod_{j \in R} r_j}
    \\ &
    \leq   \frac{(\VV{1}{i-m}) \cdot (\VV{i-m+1}{i}) \cdot (\VV{i+1}{n})}{ 
    \VV{1}{i-m} \phantom{ )\cdot \VV{i-m+1}{i}) \cdot (\ } \VV{i+1}{n}}  
    \\ &
    =\VV{i-m+1}{i}.
  \end{align*}

  The three terms are bounded as follows.
\begin{enumerate}

\item Numerator: $\frac{r_{1}\cdots r_n}{\prod_{j \in R}r_j} \leq \det(L) = \VV{1}{n}$.

  Minimize the determinant with input matrix $\matG$ consisting of columns $\ve_i$, $i \in R$.

\item Left side of denominator is $\VV{1}{i} \leq \|\vb_1^*\|\cdots \|\vb_i^*\|
  \leq r_1 \cdots r_i, 
   \forall i \in [n].$

$\| \vc_i^* \| \leq \| \vb_i^* \|$ by LLL, and $\| \vb_i^* \| \leq r_i$ by 
Lemma~\ref{lem:MGalg}.

\item Right side of denominator:
  $\VV{i+1}{n} = \frac{1}{\|\vd_{i+1}^\dag\|} \cdots
  \frac{1}{\|\vd_{n}^\dag\|} \leq r_{i+1} \cdots r_n$, where
  $\vd_1,\ldots,\vd_n$ is the reversed dual basis of the
  dual lattice $\hat{L}$.  The first
  equality is by Corollary 6, second bullet, in Micciancio Lecture 3~\cite{Miclec3}.
  By definition, the middle quantity is 1/determinant of the
  lattice $\hat{L}_{i+1}$ generated by vectors $\vd_{i+1},\ldots, \vd_n$.
  The dual lattice
  $\hat{L} \subseteq \hat{H} := \frac{1}{r_1} {\mathbb Z} \times \cdots \times
  \frac{1}{r_n} {\mathbb Z}$ because $L$ and $H$ have the same
  span so $H\subseteq L$ if and only if $\hat{H} \supseteq \hat{L}$.  Then
  $\frac{1}{r_{i+1}}\cdots \frac{1}{r_n} \leq\det(\hat{L}_{i+1})$ because
  it is minimized by choosing the lattice $\hat{L}$ with basis
  $\frac{1}{r_1}\ve_1, \ldots, \frac{1}{r_n}\ve_n$.
\end{enumerate}

\item Let $r_{\max} = \max_{j\in R} r_ j$.  Then
$\prod_{j\in R} r_j \leq r_{\max}^\gr$.  Combine 1 and 2 to get
  \begin{align*}
    \frac{\lambda_1}{r_{\max}^{\gr/m} \Delta^{(m-1)/2}} \leq
    \frac{(r_{i-m+1}\cdots 
    r_i)^{1/m} \phantom{qqqqq} }{(\prod_{j\in R} 
      r_j)^{1/m} \Delta^{(m-1)/2}}
      \leq   \frac{(\VV{i-m+1}{i})^{1/m}}{\Delta^{(m-1)/2}}
      \leq  \|\vc_i^*\|,
  \end{align*}
  where the first inequality follows from the fact that for all $m\leq i \leq n$, 
$\lambda_1^m \leq \min \{r_{i-m+1}, \ldots, r_i\}^m \leq
r_{i-m+1}\cdots r_i$ because $\lambda_1\leq r_i$, for all $i$.

It remains to choose the
smallest $m$ so that the bound in case 2 is at least the bound in case
1, where $i\leq m$, i.e., so that
$\frac{\lambda_1}{\Delta^{m-1}} \leq \frac{\lambda_1}{r_{\max}^{\gr/m}
  \Delta^{(m-1)/2}}$.  This is optimized if
$r^{\gr/m}_{\max} \approx \Delta^{(m-1)/2}$, or $ \frac{\gr}{m} \ln
r_{\max} \approx \frac{m-1}{2} \ln \Delta$, so choose $m \approx
\sqrt{\gr \ln q \ln \Delta}
\geq 
\sqrt{\gr\ln r_{\max} \ln \Delta}$.
\end{enumerate}
\end{proof}

\begin{theorem}
  There is an algorithm running in time $\poly(n, \log q)$ that takes
  a lattice $L\subseteq \Z^n$ and a vector $\vt \in \Z^n$, and solves
  BDD with approximation factor $\exp(-\Omega(\sqrt{r\ln q}))$.

  Furthermore, SIVP in $L^\perp = qL^*$ can be solved with
  length $q \cdot\exp(O(\sqrt{r \ln q}))/\lambda_1(L)$.
\end{theorem}

\begin{proof}
  Babai's BDD algorithm works up to a decoding radius of $\min_i  
  \|\vc^*_i\|$, which is lower bounded in Lemma~\ref{lem:mingslen}.

  For computing short vectors,~\cite[Lemma 1.3]{GPV08} gives an upper
  bound on the smoothing radius of $L^\perp$ by the max basis length
  of the Gram-Schmidt vectors in the dual.  This is equal to the
  minimum basis length in the Gram-Schmidt vectors in the primal,
  which by Lemma~\ref{lem:mingslen} gives a bound on the lengths of $q\cdot
  \exp(O(\sqrt{\gr \ln q}))/\lambda_1(L)$. 
\end{proof}

Sampling short vectors can also be done using a quantum algorithm
using BDD to create phaseless Gaussian superpositions so that
computing the Fourier transform and measuring samples lattice points
from a discrete Gaussian of a similar length.

{\bf Acknowledgements:}  Thanks to Kirsten Eisentraeger for valuable
help.  This paper would probably not be possible without the lecture
notes of Daniele Micciancio and the lecture notes of Vinod Vaikuntanathan.
Thanks also to L\'{e}o Ducas and Chris Peikert for helpful conversations.


\newcommand{\etalchar}[1]{$^{#1}$}

\end{document}